\documentclass[runningheads]{llncs}
\usepackage{amsfonts}
\usepackage{amsmath}
\usepackage{tikz-cd}
\usepackage{tikzscale}

\usepackage{amsthm}
\usepackage{centernot}
\usepackage{graphicx}
\usepackage[caption=false]{subfig}
\usepackage{algorithm}
\usepackage[noend]{algpseudocode}
\usepackage{multirow}

\floatname{algorithm}{Network}

\newcommand{\tuple}[1]{\langle #1 \rangle}
\newcommand{\nat}{\mathbb{N}}
\newcommand{\bool}{\{0,1\}}
\newcommand{\zip}{\textrm{zip}}
\newcommand{\odd}{\textit{odd}}
\newcommand{\even}{\textit{even}}
\newcommand{\floor}[1]{\left \lfloor #1 \right \rfloor }
\newcommand{\ceil}[1]{\left \lceil #1 \right \rceil }
\newcommand{\snd}{\textrm{2nd}}
\newcommand{\trd}{\textrm{3rd}}
\newcommand{\sort}{\textit{sort}}
\newcommand{\pref}{\textrm{pref}}
\newcommand{\suff}{\textrm{suff}}
\newcommand{\remove}{\textrm{drop}}

\newtheorem{observation}{Observation}

\title{Encoding Cardinality Constraints using Generalized Selection Networks}

\titlerunning{Encoding Cardinality Constraints using Generalized Selection Networks}

\author{
Micha\l~Karpi\'nski, Marek Piotr\'ow}

\authorrunning{Micha\l~Karpi\'nski \and  Marek Piotr\'ow}

\date{}

\institute{Institute of Computer Science, University of Wroc\l aw\\
  Joliot-Curie 15, 50-383 Wroc\l aw, Poland\\
\email{\{karp,mpi\}@cs.uni.wroc.pl}}

\begin{document}

\pagestyle{headings}
%In order to omit page numbers and running heads
%please change this line to
%\pagestyle{empty}
%and change the first command line too, see above.

\maketitle

\begin{abstract}
  \noindent Boolean cardinality constraints state that at most (at least, or exactly) $k$ out
  of $n$ propositional literals can be true. We propose a new class of selection
  networks that can be used for an efficient encoding of them. Several comparator
  networks have been proposed recently for encoding cardinality constraints and experiments have
  proved their efficiency. Those were based mainly on the odd-even or pairwise
  comparator networks. We use similar ideas, but we extend the model of comparator networks so that
  the basic components are not only comparators (2-sorters) but more general $m$-sorters, for $m \geq 2$.  
  The inputs are organized into $m$ columns, in which elements are recursively selected
  and, after that, columns are merged using an idea of multi-way merging.
  We present two algorithms parametrized by $m \geq 2$. We call those networks
  $m$-Wise Selection Network and $m$-Odd-Even Selection Network.
  We give detailed construction of the mergers when $m=4$.
  The construction can be directly applied to any values of $k$ and $n$. The proposed
  encoding of sorters is standard, therefore the arc-consistency is preserved. We
  prove correctness of the constructions and present the theoretical and
  experimental evaluation, which show that the new encodings are competitive
  to the other state-of-art encodings.
\end{abstract}

\section{Introduction}

Several hard decision problems can be efficiently reduced to Boolean satisfiability (SAT)
problem and tried to be solved by recently-developed SAT-solvers. Some of them are
formulated with the help of different high-level constraints, which should be either
encoded into CNF formulas or solved inside a SAT-solver by a specialized extension. There was
much research on both of these approaches.

In this paper we consider encodings of Boolean cardinality constraints that take the form
$x_1 + x_2 + \dots + x_n \sim k$, where  $x_1, x_2, \dots, x_n$ are Boolean literals (that
is, variables or their negations) and $\sim$ is a relation from the set $\{<, \le, = , \ge
, >\}$. Such cardinality constraints appear naturally in formulations of different
real-world problems including cumulative scheduling \cite{sfsw09}, timetabling
\cite{aan12} or formal hardware verification \cite{bccz99}.

In a direct encoding of a cardinality constraint $x_1 + x_2 + \dots + x_n < k$ one can take
all subsets of $X = \{x_1, \dots, x_n\}$ of size $k$ and for each of them construct a CNF
formula that states that at least one of the literals in the subset must be false. The
direct encoding is quite efficient for very small values of $k$ and $n$, but for larger
parameters another approach should be used.

\vspace{0.1cm}

\noindent {\bf Related work:} In the last years several selection networks
were proposed for the encoding cardinality constraints and experiments proved their efficiency. They were based
mainly on the odd-even or pairwise comparator networks. Codish and Zazon-Ivry
\cite{card2} introduced pairwise selection networks that used the concept of Parberry's
pairwise sorting network \cite{parberry}. Ab\'io, Nieuwenhuis, Oliveras and
Rodr\'iguez-Carbonell \cite{abio,card4} defined encodings that implemented the odd-even
sorting networks by Batcher \cite{batcher}. In \cite{abio} the authors proposed a mixed
parametric approach to the encodings, where the direct encoding is chosen for small
sub-problems and the splitting point is optimized when large problems are divided into two
smaller ones. They proposed to minimize the function $\lambda \cdot num\_vars +
mun\_clauses$ in the encodings. The constructed encodings are small and efficient,
but non-trivial extra time is needed to find the optimal splitting points.

\vspace{0.1cm}

\noindent {\bf Our contribution:} Partially inspired by this
optimization criterion we started looking for selection
networks that can be encoded with a smaller number of auxiliary variables and, if possible,
not much larger number of clauses, and where there is no need for the costly optimization.
In addition, we investigate the influence of our encodings on the execution times of SAT-solvers
to be sure that the new algorithms can be used in practice.
The obtained constructions are presented in this paper. The main idea is to split the
problem into $m$ sub-problems, recursively select elements in them and then
merge the selected subsequences using an idea of multi-way merging \cite{leeb,china}. In such
approach we can use not only comparators (2-sorters) but also $m$-sorters,
which can be encoded directly.

Using this generalized version of comparators we have
created two novel classes of networks parametrized by $m \geq 2$. First one is called
{\bf $m$-Wise Selection Network}, in which we use the idea of pairwise sorting \cite{parberry}
combined with a multi-way merging algorithm by Gao, Liu and Zhao \cite{china},
but in more general way. The second class is called {\bf $m$-Odd-Even Selection Networks},
where the multi-way merge sorting networks by Batcher and Lee \cite{leeb} are generalized
in a way that we can recursively select $k$ largest elements from each of the $m$ sub-problems.
It should be also pointed out that none of the algorithms presented in this paper is using
or is even based on the construction presented in the CP'15 paper of Karpi\'nski and Piotr\'ow \cite{kp15}.
The reason is that there the size of an input sequence is required to be a power of 2.

Both algorithms are presented using divide and conquer paradigm. The key to achieve
efficient algorithms lie in the construction of networks that combines the
results obtained from the recursive calls. The construction of those {\em mergers}
is the main result of this paper. We fix $m=4$ and give detailed constructions
for two merging networks: $4$-Wise Merging Network and $4$-Odd-Even Merging Network.
We prove their correctness and compare the numbers of variables and
clauses of the corresponding encodings to their two counterparts:
Pairwise Merging Networks and Odd-Even Merging Networks \cite{card2}. The
calculations show that encodings based on our merging networks uses less number of
variables and clauses, when $k<n$. In case of $4$-Odd-Even Merging Network this
also translates into smaller encodings in terms of variables and clauses for the
$4$-Odd-Even Selection Network in comparison to Odd-Even Selection Network,
for sufficiently large $k$ and $n$.

The constructions are parametrized by
any values of $k$ and $n$, so they can be further optimized by mixing them with other
constructions. For example, in our experiments we mixed them with the direct encoding for
small values of parameters. We used the standard encoding of $m$-sorters, therefore the
arc consistency is preserved. Finally, we present results of our experiments, which is
another main contribution of this paper. We show that generalized selection networks
are superior to standard selection networks previously proposed in the literature, in context
of translating cardinality constraints into propositional formulas. We also conclude that
although encodings based on pairwise approach use less number of variables than odd-even encodings,
in practice, it is the latter that achieve better reduction in SAT-solving runtime.
It is a helpful observation, because from the practical point of view, implementing odd-even
networks is a little bit easier.

We also empirically compare our encodings with other state-of-art encodings, not only based on
comparator networks, but also on binary adders and binary decision diagrams. Those are mainly
used in encodings of Pseudo-Boolean constraints, but it is informative to see how well they
perform when encoding cardinality constraints.

\vspace{0.1cm}

\noindent {\bf Structure of the paper:} The rest of the paper is organized
as follows: Section \ref{sec:pre} contains definitions and
notations used in the paper. In Section \ref{sec:sel} the construction of the generalized
selection networks and the proofs of correctness are given. Section \ref{sec:merge} is
devoted to specialized merging networks for $m=4$ and contains detailed proofs of correctness.
In Section \ref{sec:comp} we compare the encodings
produced by the construction with Pairwise Cardinality Network and Odd-Even Selection Network, in
terms of number of variables and clauses each of them uses. Our
experimental evaluation is presented in Section \ref{sec:exp} followed by Conclusions.

\section{Preliminaries} \label{sec:pre}

In this section we introduce definitions and notations used in the rest of the paper. Let
$X$ denote a totally ordered set, for example the set of natural numbers $\nat$ or the set
of binary values $\{0, 1\}$. We introduce the auxiliary "smallest" element
$\bot \notin X$, such that for all $x \in X$ we have $\bot < x$.
Thus, $X \cup \{\bot\}$ is totally ordered. The element $\bot$ is used in the next section
to simplify presentation of algorithms.

  \begin{definition}[sequences]
  A sequence of length $n$, say $\bar{x} = \tuple{x_1, \dots, x_n}$, is an element of
  $X^n$. In particular, an element of $\{0, 1\}^n$ is called a {\em binary} sequence. We
  say that a sequence $\bar{x} \in X^n$ is {\em sorted} if $x_i \geq x_{i+1}$, $1 \le i <
  n$. Given two sequences $\bar{x} = \tuple{x_1, \dots, x_n}$ and $\bar{y} = \tuple{y_1,
  \dots, y_m}$ we define {\em concatenation} as $\bar{x} :: \bar{y} = \tuple{x_1, \dots,
  x_n, y_1, \dots, y_m}$. We use also the following subsequence notation: $\bar{x}_{\odd}
  = \tuple{x_1, x_3, \dots}$, $\bar{x}_{\even} = \tuple{x_2, x_4, \dots}$,
  $\bar{x}_{a,\dots,b} = \tuple{x_a, \dots, x_b}$, $1 \le a \leq b \le n$, and the {\em
  prefix/suffix} operators: $pref(i,\bar{x}) ~=~ \bar{x}_{1,\dots,i}$ and $suff(i,\bar{x}) ~=~
  \bar{x}_{i,\dots,n}$, $1 \leq i \leq n$. The number of a given value $b$ in $\bar{x}$ is
  denoted by $|\bar{x}|_b$ and the result of removing all occurrences of $b$ in $\bar{x}$ 
  is written as $\remove(b,\bar{x})$.

  Let $n,m \in \nat$. We define a {\em domination} relation '$\succeq$' on $X^n \times
  X^m$. Let $\bar{x} \in X^n$ and $\bar{y} \in X^m$, then: $\bar{x} \succeq \bar{y} \iff
  \forall_{i \in \{1,\dots,n\}} \forall_{j \in \{1,\dots,m\}} \; x_i \geq y_j$
  \end{definition}

  \begin{definition}[zip operator]
  For $m \ge 1$ given sequences (column vectors) $\bar{x}^i = \tuple{x^i_1, x^i_2, \dots,
  x^i_{n_i}}$, $1 \le i \le m$ and $n_1 \ge n_2 \ge \dots \ge n_m$, let us define the
  $\zip$ operation that outputs the elements of the vectors in row-major order:
    \[
    \zip(\bar{x}^1, \dots, \bar{x}^m) = \begin{cases}
    \bar{x}^1 & \mbox{if~} m = 1 \\ 
    \zip(\bar{x}^1, \dots, \bar{x}^{m-1}) & \mbox{if~} |\bar{x}^m| = 0 \\ 
    \tuple{x^1_1, x^2_1, \dots, x^m_1} ::
    \zip(\bar{x}^1_{2,\dots, n_1}, \dots, \bar{x}^m_{2, \dots, n_m}) & \mbox{otherwise}
    \end{cases}
  \]
  \end{definition}
  
  We construct and use comparator networks in this paper. The standard definitions and
  properties of them can be found, for example, in \cite{knuth}. The only difference is
  that we assume that the output of any sorting operation or comparator is in a
  non-increasing order.
  
  \begin{definition}[m-sorter]\label{def:msort}
    A comparator network $f^m$ is a {\em sorting network} (or {\em m-sorter}), if for each
    $\bar{x} \in X^m$, $f^m(\bar{x})$ is sorted.
  \end{definition}

  In the pairwise sorting network, the first step is to split the input into two equally
  sized sequences $\bar{x}$ and $\bar{y}$, such that for any $i$, $x_i \geq y_i$.
  Then $\bar{x}$ and $\bar{y}$ are sorted independently and finally merged. Even after sorting, the property
  $x_i \geq y_i$ is maintained (see \cite{parberry}). For example,
  the pair of sequences $\tuple{110,100}$ are pairwise sorted. We would like to extend this notion
  to cover larger number of sequences, let's say $m \geq 2$. For example, the tuple
  $\tuple{111,110,100,000}$ is 4-wise (sorted). The following definition captures this idea.
  
\begin{definition}[m-wise sequences]\label{def:mws}
  Let $c, k, m \in \nat$, $1 \le k$ and $k/m \le c$. Moreover, let $k_i=\min(c,
  \floor{k/i})$ and $\bar{x}^i \in X^{k_i}$, $1 \leq i \leq m$. The tuple
  $\tuple{\bar{x}^1, \dots, \bar{x}^m}$ is $m$-wise of order $(c,k)$ if:
  \begin{enumerate}
    \item $\forall_{1 \leq i \leq m}$ $\bar{x}^i$ is sorted,
    \item $\forall_{1 \leq i \leq m-1}$ $\forall_{1 \leq j \leq n_{i+1}}$ $x^i_j \geq x^{i+1}_j$.
  \end{enumerate}
  \end{definition}
  
  \begin{observation} \label{obs:mws}
    Let a tuple $\tuple{\bar{x}^1, \dots, \bar{x}^m}$ be $m$-wise of order 
    $(c,k)$. Then:
    \begin{enumerate}
      \item $|\bar{x}^1| \ge |\bar{x}^2| \ge  \dots \ge |\bar{x}^m| = \floor{k/m}$,
      \item $\sum_{i=1}^{m} |\bar{x}^i| \ge k$,
      \item if $|\bar{x}^{i-1}| > |\bar{x}^i| + 1$ then $|\bar{x}^i| = 
      \floor{k/i}$. 
    \end{enumerate}
  \end{observation}
  
  \begin{proof}
    The first statement is obvious. To prove the second one let $k_i = |\bar{x}^i|$, $1
    \le i \le m$. Then, if $k_i=c$ for each $1 \leq i \leq m$, we have $\sum_{i=1}^{m} k_i
    = m c \geq m k/m = k$ and we are done. Let $1 \leq i \leq m$ be the first index such
    that $k_i\neq c$, therefore $k_i=\floor{k/i} < c$. Thus, for each $1 \leq j < i$: $k_j
    = c > k_i$, therefore $k_j \geq k_i + 1$. From this we get that $\sum_{j=1}^{m} k_j
    \geq \sum_{j=1}^{i} k_j \geq i\floor{k/i} + i - 1 \geq k$.

    The third one can be easily proved by contradiction. Assume that $k_i = c$, then from
    the first property $k_{i-1} \geq k_i = c$. Therefore, $k_{i-1} = \min(c,
    \floor{k/(i-1)}) \geq c$, so $k_{i-1}=c$, a contradiction.
  \end{proof}
  
  \begin{definition}[top $k$ sorted sequence]
    A sequence $\bar{x} \in X^n$ is top $k$ sorted, with $k \leq n$, if 
    $\tuple{x_1,\dots,x_k}$ is sorted and
    $\tuple{x_1,\dots,x_k} \succeq \tuple{x_{k+1},\dots,x_n}$.
  \end{definition}
  
  \begin{definition}[m-wise merger]\label{def:mwm}
    A comparator network $f^s_{(c,k)}$ is a {\em m-wise merger of order $(c,k)$},
    if for each $m$-wise tuple $T = \tuple{\bar{x}^1, \dots, \bar{x}^m}$ of order
    $(c,k)$, such that $s = \sum_{i=1}^{m} |\bar{x}^i|$, $f^s_{(c,k)}(T)$ is top
    $k$ sorted.
  \end{definition}

  \begin{definition}[m-odd-even merger]\label{def:oem}
    A comparator network $f^s_{k}$ is a {\em m-odd-even merger of order $k$},
    if for each tuple $T = \tuple{\bar{x}^1, \dots, \bar{x}^m}$,
    where each $\bar{x}^i$ is top $k$ sorted and $s = \sum_{i=1}^{m} |\bar{x}^i|$,
    $f^s_k(T)$ is top $k$ sorted.
  \end{definition}
  
  \begin{definition}[selection network]
    A comparator network $f^n_k$ (where $k \leq n$) is a {\em $k$-selection network},
    if for each $\bar{x} \in X^n$, $f^n_k(\bar{x})$ is top $k$ sorted.
  \end{definition}

  A clause is a disjunction of literals (Boolean variables $x$ or their negation $\neg
  x$). A CNF formula is a conjunction of one or more clauses. Cardinality constraints are
  of the form $x_1 + \dots + x_n \sim k$, where $k \in \nat$ and $\sim$ belongs to
  $\{<,\leq,=,\geq,>\}$. We will focus on cardinality constraints with less-than relation,
  i.e. $x_1 + \dots + x_n < k$. The other can be easily translated to such form (see \cite{card4}).
  
  In \cite{abio,card4,card2,card3} authors are using sorting networks for an encoding of
  cardinality constraints, where inputs and outputs of a comparator are Boolean variables
  and comparators are encoded as a CNF formula. In addition, the $k$-th greatest output
  variable $y_k$ of the network is forced to be 0 by adding $\neg y_k$ as a clause to the
  formula that encodes $x_1 + \dots + x_n < k$. We use similar approach, but rather than
  using simple comparators (2-sorters), we also use comparators of higher order
  as building blocks. The $m$-sorter can be encoded as follows:
  for input variables $\bar{x}$ and output variables $\bar{y}$, we add the set of clauses
  $\{x_{i_1} \wedge \dots \wedge x_{i_p} \Rightarrow y_p \, : \, 1 \leq p \leq m, 1 \leq
  i_1 < \dots < i_p \leq m\}$. Therefore, we need $m$ auxiliary variables and $2^m-1$
  clauses.
  
\section{Generalized Selection Networks} \label{sec:sel}

Here we present two novel constructions for selection networks which uses sorters
of arbitrary size as components. We want to apply our algorithms for CNF encoding,
therefore the only non-trivial operation that we are allowed
to use in generalized comparator networks is $sort^m : X^m \rightarrow X^m$,
which is a $m$-sorter. Here (and in the rest of the paper)
we use it as a black box, but keep in mind that in the actual implementation one should encode
the $m$-sorter using the standard procedure explained at the end of Section \ref{sec:pre}.
We would also like to note that only first $k$ output variables for each selection network is
of interest, but to stay consistent with the definitions of Section \ref{sec:pre}, we
write our algorithms so that the size of the output is the same as the size of the input. To this
end we introduce the variable $\overline{out}$ in which we store all throw-away variables in arbitrary order.
The sequence $\overline{out}$ is then appended to the output of the algorithms.

\subsection{m-Wise Selection Network}

We begin with the algorithm for constructing the $m$-Wise Selection Network
(Network \ref{net:mw_sel}) and we prove that it is correct. In Network \ref{net:mw_sel} we
use $mw\_merge^{s}_{(c,k)}$, that is, a $m$-Wise Merger of order $(c,k)$, as a block box.
We give detailed constructions of $m$-Wise Merger for $m=4$ in the next section.

The idea we use is the generalization of the one used in Pairwise Cardinality Network from
\cite{card2} which is based on the Pairwise Sorting Network by Parberry \cite{parberry}.
First, we split the input sequence into $m$ columns of non-increasing sizes (lines 2--5)
and we sort rows using sorters (lines 6--8). Then we recursively run the selection
algorithm on each column (lines 9--11), where at most $\floor{k/i}$ items are selected
from the $i$-th column. In obtained outputs, selected items are sorted and form prefixes of the columns. The
prefixes are padded with zeroes (with $\bot$'s) in order to get the input sizes required by the $m$-wise
property (Definition \ref{def:mws}) and, finally, they are passed to the merging procedure
(line 12--13). The base case, when $k=1$, is handled by the auxiliary network $max^n$, which
outputs the maximum of $n$ elements. The standard construction of $max^n$ uses $n-1$ 2-sorters.

\begin{example}\label{ex:mw_sel}
  In Figure \ref{fig:mw_sel_example} we present a sample run of Network \ref{net:mw_sel}. The input is
  a sequence $1111010010000010000101$, with parameters $\tuple{n_1,n_2,n_3,n_4,k}=\tuple{8,7,4,3,6}$.
  First step (Figure \ref{subfig:mw:a}) is to arrange the input in columns (lines 2--5). In this example
  we get $\bar{x}^1=\tuple{11110100}$, $\bar{x}^2=\tuple{1000001}$, $\bar{x}^3=\tuple{0000}$
  and $\bar{x}^4=\tuple{101}$. Next we sort rows using $1,2,3$ or $4$-sorters (lines 6--8), the result is
  visible in Figure \ref{subfig:mw:b}. We make recursive calls in lines 9--11 of the algorithm. Items selected
  recursively in this step are marked on Figure \ref{subfig:mw:c}. Notice that in $i$-th column we only need
  to select $\floor{k/i}$ largest elements. This is because of the initial sorting of the rows. Next comes
  the merging step, which is selecting $k$ largest elements from the results of the previous step. How
  exactly those elements are obtained and outputted depends on the implementation. We choose the convention that the
  resulting $k$ elements must be placed in the major-row order in our column representation of the input
  (see Figure \ref{subfig:mw:d}).
\end{example}

\begin{figure}
  \centering
  \subfloat[split into columns\label{subfig:mw:a}]{\includegraphics[width=0.22\textwidth]{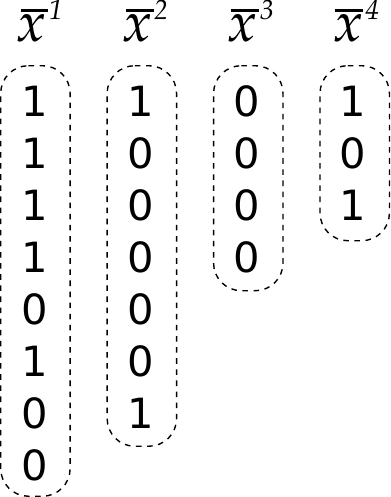}}~~
  \subfloat[sort rows\label{subfig:mw:b}]{\includegraphics[width=0.22\textwidth]{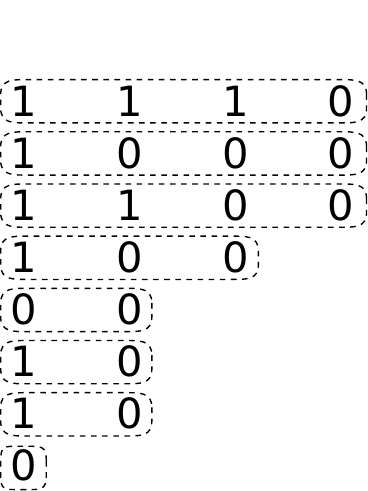}}~~
  \subfloat[select recursively\label{subfig:mw:c}]{\includegraphics[width=0.22\textwidth]{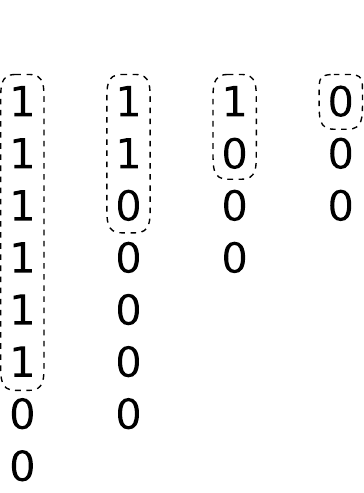}}~~
  \subfloat[merge columns\label{subfig:mw:d}]{\includegraphics[width=0.22\textwidth]{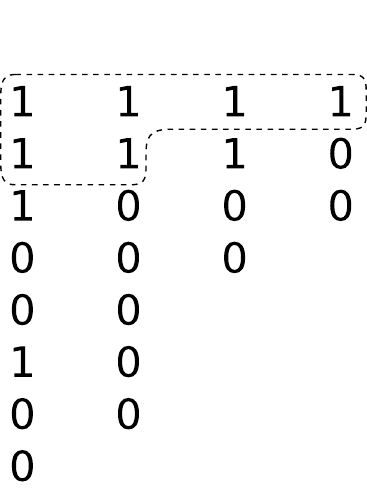}}
  \caption{Sample run of the 4-Wise Slection Network.}
  \label{fig:mw_sel_example}
\end{figure}

  \begin{algorithm}
    \caption{$mw\_sel^n_k$}\label{net:mw_sel}
    \begin{algorithmic}[1]
      \Require {$\bar{x} \in X^n$; $n_1, \dots , n_m \in \nat$ 
        where $n > n_1 \geq \dots \geq n_m$ and $\sum n_i = n$; $1 \le k \le n$}
      \If {$k = 1$}
        \Return $max^n(\bar{x})$
      \EndIf 
      \State $offset = 1$
      \ForAll {$i \in \{1,\dots,m\}$} \Comment{Splitting the input into columns.}
        \State $\bar{x}^i \gets \tuple{x_{offset},\dots,x_{offset + n_i - 1}}$
        \State $offset += n_i$
      \EndFor
      \ForAll {$i \in \{1,\dots,n_1\}$} \Comment{Sorting rows.}
        \State $m' = \max\{j : n_j \geq i\}$
        \State $\bar{y}^i \gets \sort^{m'}(\tuple{x^1_{i},\dots,x^{m'}_{i}})$
      \EndFor
      \ForAll {$i \in \{1,\dots,m\}$} \Comment{Recursively selecting items in columns.}
        \State $k_i=\min(n_1,\floor{k/i})$; ~~$l_i=\min(n_i,\floor{k/i})$
        \Comment{$k_i \ge l_i$}
        \State $\bar{z}^i \gets mw\_sel^{n_i}_{l_i}(\bar{y}^1_i,\dots,\bar{y}^{n_i}_i)$
      \EndFor
      \State $s = \sum_{i=1}^{m} k_i$; ~~$c = k_1$;
        ~~$\overline{out} = \suff(l_1+1,\bar{z}^1) :: \dots :: \suff(l_m+1,\bar{z}^m)$
      \State $\overline{res} \gets mw\_merge^{s}_{(c,k)}(\tuple{
      pref(l_1,\bar{z}^1)::\bot^{k_1-l_1}, \dots, pref(l_m,\bar{z}^m)}::\bot^{k_m-l_m})$
      \State \Return $\remove(\bot,\overline{res}) :: \overline{out}$
    \end{algorithmic}
  \end{algorithm}

In the following we will prove the correctness of Network \ref{net:mw_sel} using the well-known 0-1 
Principle, that is, we will consider binary sequences only.

  \begin{lemma}\label{lma:mw_num}
    Let $\bar{t}^i = \tuple{\bar{y}^1_i,\dots,\bar{y}^{n_i}_i}$, where $\bar{y}^j$, $j=1,
    \dots, n_1$, is the result of Step 8 in Network \ref{net:mw_sel}. For each $1 \leq
    i < m$: $|\bar{t}^i|_1 \geq |\bar{t}^{i+1}|_1$.
  \end{lemma}

  \begin{proof}
    Take any $1 \leq i < m$. Consider element $y^j_i$ (for some $1 \leq j \leq n_i$). Since $\bar{y}^j$ is sorted, we have
    $y^j_i \geq y^j_{i+1}$, therefore if $y^j_{i+1}=1$ then $y^j_i=1$. Thus $|\bar{t}^i|_1 \geq |\bar{t}^{i+1}|_1$.
  \end{proof}

  \begin{corollary}\label{crly:mw_num}
    For each $1 \leq i \leq m$, let $\bar{z}^i$ be the result of the step 11 in Network \ref{net:mw_sel}.
    Then for each $1 \leq i < m$: $|\bar{z}^i|_1 \geq |\bar{z}^{i+1}|_1$.
  \end{corollary}

  \begin{proof}
    For $1 \leq i \leq m$, let $\bar{t}^i$ be the same as in Lemma \ref{lma:mw_num}.
    Take any $1 \leq i < m$. Comparator network only permutate its input, therefore
    $|\bar{z}^i|_1 = |mw\_sel^{n_i}_{s_i}(\bar{t}^i)|_1 = |\bar{t}^i|_1 \geq
    |\bar{t}^{i+1}|_1 = |mw\_sel^{n_{i+1}}_{s_{i+1}}(\bar{t}^{i+1})|_1 = |\bar{z}^{i+1}|_1$.
    The inequality comes from Lemma \ref{lma:mw_num}.
  \end{proof}

  \begin{theorem}\label{thm:mw_sel}
    $mw\_sel^n_k$ is a $k$-selection network.
  \end{theorem}

  \begin{proof}
    We prove by induction that for each $n,k \in \nat$ such that $1 \leq k \leq n$ and
    each $\bar{x} \in \bool^n$: $mw\_sel^n_k(\bar{x})$ is top $k$ sorted. If $1=k \leq n$
    then $mw\_sel^n_k = max^n$, so the theorem is true. For the induction step assume that
    $n \ge k \ge 2$, $m \ge 2$ and for each $(n^{*},k^{*}) \prec (n,k)$ (in
    lexicographical order) the theorem holds. We have to prove the following two properties:
    \begin{enumerate}
      \item The tuple $\tuple{pref(l_1,\bar{z}^1) :: \bot^{k_1-l_1}, \dots, pref(l_m,\bar{z}^m) :: 
      \bot^{k_m-l_m}}$ is $m$-wise of order $(k_1, k)$.
      \item The sequence $\bar{w} = pref(l_1,\bar{z}^1)::\dots ::pref(l_m,\bar{z}^m)$ contains 
      $k$ largest elements from $\bar{x}$.
    \end{enumerate}
    
    Ad. 1): Observe that for any $i$, $1 \leq i \leq m$, we have $n_i \leq n_1 < n$ and
    $l_i \leq l_1 \leq k$. Thus, $(n_i,l_i) \prec (n,k)$ and $\bar{z}_i$ is top $l_i$
    sorted due to the induction hypothesis. Therefore, $pref(l_i,\bar{z}^i)$ is sorted and
    so is $pref(l_i,\bar{z}^i) :: \bot^{k_i-l_i}$. In this way, we prove that
    the first property of Definition \ref{def:mws} is satisfied. To prove the second one,
    fix $1 \leq j \leq l_{i+1}$ and assume that $z^{i+1}_j = 1$. We are going to
    show that $z^i_j = 1$. Since $\bar{z}^i$ and $\bar{z}^{i+1}$ are both top $l_{i+1}$
    sorted (by the induction hypothesis) and since $|\bar{z}^i|_1 \geq |\bar{z}^{i+1}|_1$
    (from Corollary \ref{crly:mw_num}), we have $|pref(l_{i+1},\bar{z}^i)|_1 \geq
    |pref(l_{i+1},\bar{z}^{i+1})|_1$, so $z^i_j=1$.
    
    Ad. 2): It is easy to observe that $\bar{z} = \bar{z}^1 :: \dots :: \bar{z}^m$ is a
    permutation of the input sequence $\bar{x}$. If all 1's in $\bar{z}$ are in $\bar{w}$,
    we are done. So assume that there exists $z^i_j=1$ for some $1 \leq i \leq m$, $l_i <
    j \leq n_i$. We will show that $|\bar{w}|_1 \geq k$. From the induction hypothesis we
    get $pref(l_i,z^i) \succeq \tuple{z^i_j}$, which implies that $|pref(l_i,z^i)|_1 =
    l_i$. From (1) and the second property of Definition \ref{def:mws} it is clear that
    all $z^s_t$ are 1's, where $1 \leq s \leq i$ and $1 \leq t \leq l_i$, therefore
    $|pref(l_i,\bar{z}^1):: \dots :: pref(l_i,\bar{z}^i)|_1 = i \cdot l_i$. Moreover,
    since $j > l_i$, we have $l_i=\floor{k/i}$; otherwise we would have $j >n_i$. If
    $i=1$, then $|\bar{w}|_1 \geq 1 \cdot l_1 = k$ and (2) holds.
    
    Otherwise, $l_1 > l_i$, so from the definition of $l_i$, $l_1 \geq \dots \geq l_i$, hence there
    exists $r \geq 1$ such that $\forall_{r < i' \leq i}$ $l_r > l_{i'} = l_i$. Notice
    that since $|pref(l_i,z^i)|_1=l_i$ and $z^i_j=1$ where $j>l_i$ we get $|z^i|_1 \geq
    l_i + 1$. From Corollary \ref{crly:mw_num} we have that for $1 \leq r' \leq r$:
    $|\bar{z}^{r'}|_1 \geq l_i + 1$. Using $l_{r'} \geq l_r > l_i$ and the induction
    hypothesis we get that each $\bar{z}^{r'}$ is top $l_{i}+1$ sorted, therefore
    $|pref(l_{i}+1,\bar{z}^{r'})|_1 = l_i+1$. We finally have that $|\bar{w}|_1 \geq
    |pref(l_{i}+1,\bar{z}^1):: \dots :: pref(l_{i}+1,\bar{z}^r) :: pref(l_i,\bar{z}^{r+1})
    :: \dots :: pref(l_i,\bar{z}^{i})|_1 = r (l_{i}+1) + (i-r)l_i = r (l_{r+1}+1) +
    (i-r)l_{r+1} = i \cdot l_{r+1} + r \geq (r+1)\floor{k/(r+1)} + r \geq k$. In the
    second to last inequality, we use the facts: $l_{r+1}=l_i=\floor{k/i} < n_i \leq n_{r+1}$
    from which $l_{r+1}=\floor{k/(r+1)}$ follows.

    From the statements (1) and (2) we can conclude that $mw\_merge^{s}_{(n,k)}$ will return 
    the $k$ largest elements from $\bar{x}$, which completes the proof.
  \end{proof}

  \subsection{m-Odd-Even Selection Network}
  
  Next we present the algorithm for constructing the $m$-Odd-Even Selection Network
  (Network \ref{net:oe_sel}) and we prove that it is correct. In Network \ref{net:oe_sel} we
  use $oe\_merge^{s}_{k}$ as a black box. It is a $m$-Odd-Even Merger of order $k$.
  We give detailed constructions of $m$-Odd-Even Merger for $m=4$ in the next section.

  The idea we use is the generalization of the one used in Odd-Even Selection Network from
  \cite{card2} which is based on the Odd-Even Sorting Network by Batcher \cite{batcher}.
  We arrange the input sequence into $m$ columns of non-increasing sizes and we recursively
  run the selection algorithm on each column (lines 3--6), where at most $k$ items are selected
  from each column. Notice that each column is represented by ranges derived from the increasing value of variable $offset$.
  Selected items are sorted and form prefixes of the columns and they are the input to the merging procedure (line 7--8).

  \begin{example}
    Assuming the same input sequence and parameters as in Example \ref{ex:mw_sel}, the sample run of
    4-Odd-Even Selection Network would be very similar to the one presented in Figure \ref{fig:mw_sel_example},
    except we would not do the row sorting step, and in consequence we would need to recursively select $k$ largest
    elements from each column.
  \end{example}
  
  \begin{algorithm}
    \caption{$oe\_sel^n_k$}\label{net:oe_sel}
    \begin{algorithmic}[1]
      \Require {$\bar{x} \in \bool^n$; $n_1, \dots , n_m \in \nat$ where $n > n_1 \geq \dots \geq
        n_m$ and $\sum n_i = n$; $1 \le k \le n$}
      \If {$k = 1$}
        \Return $max^n(\bar{x})$
      \EndIf
      \State $offset = 1$
      \ForAll {$i \in \{1,\dots,m\}$}
        \State $k_i=\min(k, n_i)$
        \State $\bar{y}^i \gets oe\_sel^{n_i}_{k_i}(\tuple{x_{offset},\dots,x_{offset + n_i - 1}})$
        \State $offset += n_i$
      \EndFor
      \State $s = \sum_{i=1}^{m} k_i$; 
         ~~$\overline{out} = \suff(k_1+1,\bar{y}^1) :: \dots :: \suff(k_m+1,\bar{y}^m)$
      \State \Return $oe\_merge^{s}_{k}(
        \tuple{\pref(k_1,\bar{y}^1),\dots,\pref(k_m,\bar{y}^m)}) :: \overline{out}$
    \end{algorithmic}
  \end{algorithm}

  \begin{theorem}
    $oe\_sel^n_k$ is a $k$-selection network.
  \end{theorem}

  \begin{proof}
    It is easy to observe that $\bar{y} = \bar{y}^1 :: \dots :: \bar{y}^m$ is a
    permutation of the input sequence $\bar{x}$. We prove by induction that for each $n,k
    \in \nat$ such that $1 \leq k \leq n$ and each $\bar{x} \in \bool^n$:
    $oe\_sel^n_k(\bar{x})$ is top $k$ sorted. If $1=k \leq n$ then $oe\_sel^n_k = max^n$,
    so the theorem is true. For the induction step assume that $n \ge k \ge 2$, $m \ge 2$
    and for each $(n^{*},k^{*}) \prec (n,k)$ (in lexicographical order) the theorem holds.
    We have to prove that the sequence $\bar{w} = pref(k_1,\bar{y}^1)::\dots
    ::pref(k_m,\bar{y}^m)$ contains $k$ largest elements from $\bar{x}$. If all 1's in
    $\bar{y}$ are in $\bar{w}$, we are done. So assume that there exists $y^i_j=1$ for
    some $1 \leq i \leq m$, $k_i < j \leq n_i$. We will show that $|\bar{w}|_1 \geq k$.
    Notice that $k_i=k$, otherwise $j > k_i = n_i$ -- a contradiction. Since $|\bar{y}^i|
    = n_i \leq n_1 < n$, from the induction hypothesis we get that $\bar{y}^i$ is top $k_i$
    sorted. In consequence, $pref(k_i,\bar{y}^i) \succeq \tuple{y^i_j}$, which implies
    that $|pref(k_i,y^i)|_1 = k_i = k$. We conclude that $|\bar{w}|_1 \geq
    |pref(k_i,y^i)|_1 = k$. Note also that in the case $n = k$ we have all $k_i = \min(n_i,k) 
    < k$, so the case is correctly reduced.
    
    Finally, using $oe\_merge^{s}_{k}$ the algorithm returns 
    $k$ largest elements from $\bar{x}$, which completes the proof.
  \end{proof}
  
\section{Merging Networks} \label{sec:merge}

In this section we give the detailed constructions of the networks $4mw\_merge$ and
$4oe\_merge$ that merge four sequences (columns) obtained from the recursive calls in
Networks \ref{net:mw_sel} and \ref{net:oe_sel}. In the first construction we can assume
that input columns form a tuple that is 4-wise, because rows were sorted before selecting
top elements in each column. In the second one input columns are just top $k$ sorted and
rows are sorted in the final steps of the merging process. In the subsections \ref{sub:2} and  
\ref{sub:3} we present the $4$-Wise Merger and $4$-Column Odd-Even Merger, respectively. 

\subsection{4-Wise Merging Networks}\label{sub:2}

In this subsection the first merging algorithm for four columns is presented.
The input to the merging procedure is $R = \tuple{pref(n_1,\bar{y}^1), \dots,
pref(n_m,\bar{y}^m)}$, where each $\bar{y}^i$ is the output of the recursive call in
Network \ref{net:mw_sel}. The main observation is the following: since $R$ is $m$-wise, if
you take each sequence $pref(n_i,\bar{y}^i)$ and place them side by side, in columns, from
left to right, then the sequences will be sorted in rows and columns.
The goal of the networks is to put the $k$ largest
elements in top rows. It is done by sorting slope lines with decreasing slope rate, in lines 1--14
(similar idea can be found in \cite{china}).
The algorithm is presented in Network \ref{net:4mw_merge}. The pseudo-code looks non-trivial, but it
is because we need a separate sub-case every time we need to use either $\sort^2$,  $\sort^3$ or
$\sort^4$ operation, and this depends on the sizes of columns and the current slope.

After slope-sorting phase some elements might not be in the desired row-major order,
therefore the correction phase
is needed, which is the goal of the sorting operations in lines 15--18. Figure \ref{fig:4col}
shows the order relations of elements after $i=\ceil{\log n_1}$ iterations of the {\bf while} loop
(disregarding the upper index $i$, for clarity).
Observe that the order should be possibly corrected between $z_{j-1}$ and $w_{j+1}$ and then
the 4-tuples $\tuple{x_j, w_j, z_{j-1}, y_{j-1}}$ and $\tuple{x_{j+1}, w_{j+1}, z_j, y_j}$
should be sorted to get the row-major order. Lines 17--18 addresses certain corner cases of
the correction phase.

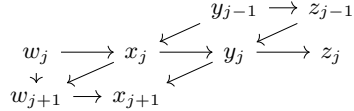
\begin{figure}
  \centering 
  \begin{tikzcd}[column sep=small,row sep=tiny]
    ~ & ~ & y_{j-1} \arrow{dl} \arrow{r} & z_{j-1} \arrow{dl}\\
    w_j \arrow{d} \arrow{r} & x_j \arrow{dl} \arrow{r} & y_j \arrow{dl} \arrow{r} & z_j \\
    w_{j+1} \arrow{r} & x_{j+1} &  & 
  \end{tikzcd}
  \caption{The order relation among elements of neighbouring rows. Arrows shows non-increasing 
    order. Relations that follow from transitivity are not shown.}
  \label{fig:4col}
\end{figure}

  An input to Network \ref{net:4mw_merge} must be 4-wise of order $(c,k)$
  and in the output should be top $k$ sorted. We prove it using the 0-1 principle,
  that is, in this subsection we assume that sequences (in particular, inputs) are binary.
  
\begin{algorithm}[!t]
\caption{$4w\_merge^s_{(c,k)}$}\label{net:4mw_merge}
\begin{algorithmic}[1]
  \Require { 4-wise tuple $\tuple{\bar{w}, \bar{x}, \bar{y}, \bar{z}}$ of order 
  $(c,k)$, where $s = |\bar{w}| + |\bar{x}| + |\bar{y}| + |\bar{z}|$ and $c = |\bar{w}|$.}
  \State $k_1 = |\bar{w}|$; ~~$k_2 = |\bar{x}|$; ~~$k_3 = |\bar{y}|$; ~~$k_4 = 
  |\bar{z}|$
  \State $d \gets \min\{l \in \nat \, | \, 2^{l} \geq k_1\}$
  \State $h \gets 2^{d}$; ~~$i \gets 0$; 
  ~~$\tuple{\bar{w}^0,\bar{x}^0,\bar{y}^0,\bar{z}^0} 
  = \tuple{\bar{w}, \bar{x}, \bar{y}, \bar{z}}$; ~~$h_0 \gets h$
  \While{$h_i > 1$} \Comment{ Define the $(i+1)$-th stage of $4w\_merge^s_{(n,k)}$.}
    \State $i \gets i + 1$; ~~$h_{i} \gets h_{i-1}/2$
    \State $\tuple{\bar{w}^i,\bar{x}^i,\bar{y}^i,\bar{z}^i} = 
      \tuple{\bar{w}^{i-1}, \bar{x}^{i-1}, \bar{y}^{i-1}, \bar{z}^{i-1}}$
    \ForAll {$j \in \{1,\dots,\min(k_3-h_i, k_4)\}$}
      \If {$j + 3h_i \leq k_1 \textbf{~and~} j + 2h_i \leq k_2$}
        $\sort^4(z^{i}_j, y^{i}_{j+h_i}, x^{i}_{j + 2h_i}, w^{i}_{j + 3h_i})$
      \ElsIf {$j + 2h_i \leq k_2$}
        $\sort^3(z^{i}_j, y^{i}_{j+h_i}, x^{i}_{j + 2h_i})$
        \Comment{$w^{i}_{j + 3h_i}$ is not defined.}
      \Else
        ~$\sort^2(z^{i}_j, y^{i}_{j+h_i})$
        \Comment{Both $x^{i}_{j + 2h_i}$ and $w^{i}_{j + 3h_i}$ are not defined.}
      \EndIf
    \EndFor
    \ForAll {$j \in \{1,\dots,\min(k_2-h_i, k_3, h_i)\}$}
      \If {$j + 2h_i \leq k_1$}
        $\sort^3(y^{i}_{j}, x^{i}_{j + h_i}, w^{i}_{j + 2h_i})$
      \Else
        ~$\sort^2(y^{i}_{j}, x^{i}_{j + h_i})$
        \Comment{$w^{i}_{j + 2h_i}$ is not defined.}
      \EndIf
    \EndFor
    \ForAll {$j \in \{1,\dots,\min(k_1-h_i, k_2, h_i)\}$}
        $\sort^2(x^{i}_{j}, w^{i}_{j + h_i})$
    \EndFor
  \EndWhile
  \Comment{Define two more stages to correct local disorders.}
  \ForAll {$j \in \{1,\dots,\min(k_1-2, k_4)\}$}
    $\sort^2(z^{i}_{j}, w^{i}_{j+2})$
  \EndFor
  \ForAll {$j \in \{1,\dots,\min(k_2-1, k_4)\}$}
    $\sort^4(y^{i}_{j}, z^{i}_{j}, w^{i}_{j+1}, x^{i}_{j+1})$
  \EndFor
  \If {$k_1 > k_4 \textbf{~and~}  k_2 = k_4$}
    $\sort^3(y^{i}_{k_4}, z^{i}_{k_4}, w^{i}_{k_4+1})$
    \Comment{$x^{i}_{k_4+1}$ is not defined.}
  \EndIf
  \If {$k \bmod 4 = 3 \textbf{~and~} k_1 > k_3$}
    $\sort^2(y^{i}_{k_4+1}, w^{i}_{k_4+2})$
    \Comment{$y^{i}_{k_4+1}$ must be corrected.}
  \EndIf
  \State \Return $\zip(\bar{w}^i,\bar{x}^i,\bar{y}^i,\bar{z}^i)$
  \Comment{Returns the columns in row-major order.}
\end{algorithmic}
\end{algorithm}

Thus, the algorithm gets as input four sorted 0-1 columns with the additional property
that the numbers of 1's in successive columns do not increase. Nevertheless, the
differences between them can be quite big. The goal of each iteration of the main
loop in Network \ref{net:4mw_merge} is to decrease the maximal possible difference
by the factor of two. Therefore, after the main loop, the differences are bounded by
one.

  \begin{example}
    A sample run of slope sorting phase of Network \ref{net:4mw_merge} is
    presented in Figure \ref{fig:mw_merge_example}. The arrows represent the sorting order.
    Notice how the 1s are being {\em pushed} towards upper-right side. In the end, the differences
    between the number of 1s in consecutive columns are bounded by one.
  \end{example}

\begin{figure}
  \centering
  \subfloat[\label{subfig:mwm:a}]{\includegraphics[width=0.22\textwidth]{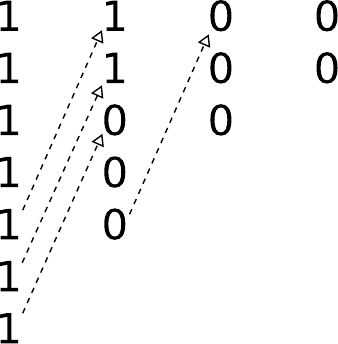}}~~
  \subfloat[\label{subfig:mwm:b}]{\includegraphics[width=0.22\textwidth]{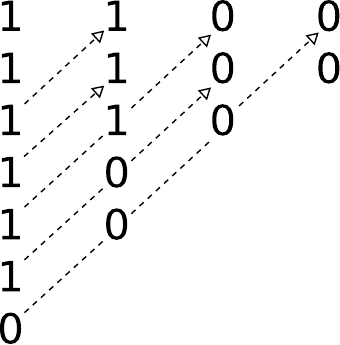}}~~
  \subfloat[\label{subfig:mwm:c}]{\includegraphics[width=0.22\textwidth]{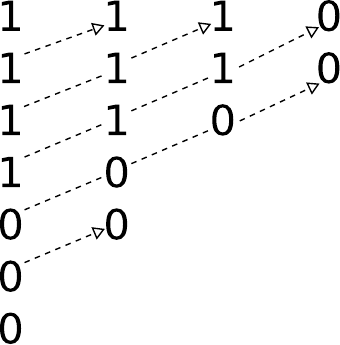}}~~
  \subfloat[\label{subfig:mwm:d}]{\includegraphics[width=0.22\textwidth]{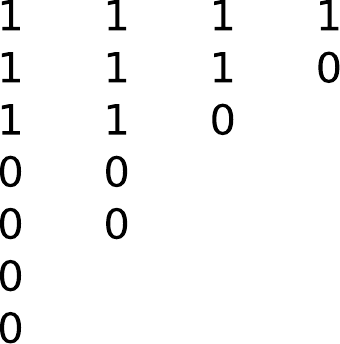}}
  \caption{Sample run of the slope-sorting phase of 4-Wise Merging Network.}
  \label{fig:mw_merge_example}
\end{figure}
  
In the $i$-th iteration of the loop the values of columns $\bar{w}^i$, $\bar{x}^i$,
$\bar{y}^i$ and $\bar{z}^i$ are obtained by sorting the results of the previous iteration
in slope lines with decreasing slope rate. The sizes of the sorted prefix of the columns
can also decrease in each iteration, but we would like to keep the $k$ largest input
elements in that prefixes. In the following observation and lemma we define and formulate the basic
properties of the column lengths after each iteration.

\begin{observation} \label{obs:abc}
  Let $k_1, \dots, k_4$ be as defined in Network \ref{net:4mw_merge}.
  For each $i$, $0 \le i \le \ceil{\log(k_1)}$, let $c_i = \min(k_3,
  \floor{k/4}+h_i)$, $b_i = \min(k_2, c_i+h_i)$ and  $a_i = \min(k_1, b_i+h_i)$.
  Then we have: $k_1 \ge a_i \ge b_i \ge c_i \ge \floor{k/4}$ and the
  inequalities: $a_i - b_i \le h_i$, $b_i - c_i \le h_i$ and $c_i - \floor{k/4}
  \le h_i$ are true.
\end{observation}

\begin{proof}
  From Observation \ref{obs:mws}.(1) we have $k_1 \ge k_2 \ge k_3 \ge
  \floor{k/4}$. It follows that $c_i \ge \floor{k/4}$ and $b_i \ge
  \min(k_3,\floor{k/4}+h_i) = c_i$ and $k_1 \ge a_i \ge \min(k_2,c_i+h_i) = b_i$.
  Moreover, one can see that $c_i - \floor{k/4} = \min(k_3-\floor{k/4},h_i) \le
  h_i$, $b_i - c_i = \min(k_2-c_i, h_i) \le h_i$ and $a_i - b_i = \min(k_1-b_i,
  h_i) \le h_i$, so we are done.
\end{proof}

\begin{lemma} \label{lem:check}
  Let $a_i, b_i, c_i$, $0 \le i \le \ceil{\log(k_1)}$, be as defined in Observation
  \ref{obs:abc}. Let $\check{w}^i = \pref(a_i,\bar{w}^i)$, $\check{x}^i =
  \pref(b_i,\bar{z}^i)$, $\check{y}^i = \pref(c_i,\bar{y}^i)$ and $\check{z}^i =
  \bar{z}^i$. Then only the items in $\check{w}^i$, $\check{x}^i$, $\check{y}^i$ and
  $\check{z}^i$ take part in the $i$-th iteration of sorting operations in lines 8-14 of
  Network \ref{net:4mw_merge}.
\end{lemma}

\begin{proof}
  An element of the vector $\bar{z}^i$, that is,  $z^i_j$ is sorted in one of the lines
  8--10 of Network \ref{net:4mw_merge}, thus for all such $j$ that $1 \le j \le k_4 =
  \floor{k/4}$. An element $y^i_j$ of $\bar{y}^i$ is sorted in lines 8--10 of the first
  inner loop and in lines 12--13 of the second inner loop. In the first three of those
  lines we have the bounds on $j$: $1 + h_i \le j \le \min(k_3,k_4+h_i)$ and in the last
  three lines - the bounds: $1 \le j \le \min(k_3,h_i)$. The sum of these two intervals
  is $1 \le j \le min(k_3, \floor{k/4}+h_i) = c_i$. Similarly, we can analyse the
  operations on $x^i_j$ in lines 8--10, 12--13 and 14. The range of $j$ is the sum of the
  following disjoint three intervals: $[1+2h_i, \min(k_2, k_3+h_i, k_4+2h_i)]$, $[1+h_i,
  \min(k_3+h_i,2h_i)]$ and $[1, \min(k_2, h_i)]$ that give us the interval $[1, \min(k_2,
  k_3+h_1, k_4+2h_i)] = [1, b_i]$ as their sum. The analysis of the range of $j$ used in
  the operations on $w^i_j$ in Network \ref{net:4mw_merge} can be done in the same way.
\end{proof}

The next lemma formulates the key invariants of the main loop in Network
\ref{net:4mw_merge}. They will be used in the proof of the main theorem in this
section stating the merging properties of the network.

\begin{lemma} \label{lma:4wsel}
  Let $\check{w}^i$, $\check{x}^i$, $\check{y}^i$ and $\check{z}^i$ be as defined in Lemma
  \ref{lem:check}. Then for all $i$, $0 \leq i \leq \ceil{\log_2 k_1}$, after the $i$-th
  iteration of the while loop in Network \ref{net:4mw_merge} the sequences $\check{w}^i,
  \check{x}^i, \check{y}^i, \check{z}^i$ are of the form $1^{p_i}0^*$, $1^{q_i}0^*$,
  $1^{r_i}0^*$ and $1^{s_i}0^*$, respectively, and:
  \begin{eqnarray}
    p_i \ge q_i \ge r_i \ge s_i, \label{eq:1}\\
    p_i - q_i \le h_i \textrm{~and~} q_i - r_i \le h_i 
      \textrm{~and~} r_i - s_i \le h_i, \label{eq:2}\\
    p_i + q_i + r_i + s_i \ge \min(k, p_0 + q_0 + r_0 + s_0). \label{eq:3}
  \end{eqnarray}
\end{lemma}
\begin{proof}
  By induction. At the beginning we have $h_0 \geq k_1 \geq k_2 \geq k3 \geq k4 =
  \floor{k/4}$ and therefore $\check{w}^0 = \bar{w}$, $\check{x}^0 = \bar{x}$,
  $\check{y}^0 = \bar{y}$ and $\check{z}^0 = \bar{z})$. All four sequences $\bar{w},
  \bar{x}, \bar{y}, \bar{z}$ are sorted (by Definition \ref{def:mws}), thus, they are of
  the form $1^{p_0}0^*$, $1^{q_0}0^*$, $1^{r_0}0^*$, and $1^{s_0}0^*$ respectively. By
  Definition \ref{def:mws}.(2), for each pair of sequences: $(\bar{w}_{1..k_2}, \bar{x})$,
  $(\bar{x}_{1..k_3}, \bar{y})$ and $(\bar{y}_{1..k_4}, \bar{z})$, if there is a 1 on the
  $j$ position in the right sequence, it must be a corresponding 1 on the same position in
  the left one. Therefore, we have: $p_0 \geq q_0 \geq r_0 \geq s_0 \geq 0$. Moreover,
  $p_0 - q_0 \leq p_0 \leq k_1 \leq h_0$, $q_0 - r_0 \leq q_0 \leq k_2 \leq h_0$ and $r_0
  - s_0 \leq r_0 \leq k_3 \leq h_0$. Finally, $p_0+q_0+r_0+s_0 \geq \min(k, p_0 + q_0 +
  r_0 + s_0)$, thus the lemma holds for $i=0$.

  In the inductive step $i > 0$ observe that the elements of $\check{w}^i$, $\check{x}^i$,
  $\check{y}^i$ and $\check{z}^i$ are defined by the sort operations over the elements
  with the same indices from vectors  $\check{w}^{i-1}$, $\check{x}^{i-1}$,
  $\check{y}^{i-1}$ and $\check{z}^{i-1}$. This means that the values of
  $w^{i-1}_{a_i+1,\dots,a_{i+1}}$, $x^{i-1}_{b_i+1,\dots,b_{i+1}}$ and
  $y^{i-1}_{c_i+1,\dots,c_{i+1}}$ are not used in the $i$-th iteration. Therefore, the
  numbers of 1's in columns that are sorted in the $i$-th iteration are defined by values:
  $p'_{i-1} = \min(a_i, p_{i-1})$, $q'_{i-1} = \min(b_i, q_{i-1})$, $q'_{i-1} = \min(c_i,
  r_{i-1})$ and $s'_{i-1} = s_{i-1}$. In the following we will prove that the numbers with
  primes have the same properties as those without them.
   \begin{eqnarray}
     p'_{i-1} \ge q'_{i-1} \ge r'_{i-1} \ge s'_{i-1}, \label{eq:4}\\
     p'_{i-1} - q'_{i-1} \le h_{i-1} \textrm{~and~} q'_{i-1} - r'_{i-1} \le h_{i-1} 
       \textrm{~and~} r'_{i-1} - s'_{i-1} \le h_{i-1}, \label{eq:5}\\
     p'_{i-1} + q'_{i-1} + r'_{i-1} + s'_{i-1} \ge \min(k, p_{i-1} + q_{i-1} + 
     r_{i-1} + s_{i-1}). \label{eq:6}
   \end{eqnarray}
  The proofs of inequalities in Eq. \ref{eq:4} are quite direct and follow from
  the monotonicity of $\min$. For example, we can observe that $p'_{i-1} =
  \min(a_i,p_{i-1}) \ge \min(b_i, q_{i-1}) = q'_{i-1}$, since we have $a_i \ge
  b_i$ and $p_{i-1} \ge q_{i-1}$, by Observation \ref{obs:abc} and the induction
  hypothesis. The others can be shown in the same way.
  
  Let us now prove one of the inequalities of Eq. \eqref{eq:5}, say, the second one.
  We have $q'_{i-1} - r'_{i-1} = \min(b_i,q_{i-1}) - \min(c_i,r_{i-1}) \le
  \min(c_i + h_i,r_{i-1} + h_{i-1}) - \min(c_i,r_{i-1}) \le h_{i-1}$, by
  Observation \ref{obs:abc}, the fact that $h_{i-1} = 2h_i$ and the induction 
  hypothesis. The proofs of the others are similar.
  
  The proof of Eq. \eqref{eq:6} is only needed if at least one of the following 
  inequalities are true: $p'_{i-1} < p_{i-1}$, $q'_{i-1} < q_{i-1}$ or $r'_{i-1} < 
  r_{i-1}$. Obviously, the inequalities are equivalent to $a_i < p_{i-1}$, 
  $b_i < q_{i-1}$ and $c_i < r_{i-1}$, respectively. Therefore, to prove Eq. 
  \ref{eq:6} we consider now three separate cases: (1) $a_i \ge p_{i-1}$ and $b_i 
  \ge q_{i-1}$ and $c_i < r_{i-1}$, (2) $a_i \ge p_{i-1}$ and $b_i < q_{i-1}$ and 
  (3) $a_i < p_{i-1}$.
  
  In the case (1) we have $p'_{i-1} = p_{i-1}$, $q'_{i-1} = q_{i-1}$ and $r'_{i-1}
  = c_i = \min(k_3,h_i + \floor{k/4}) < r_{i-1} \le q_{i-1} \le p_{i-1}$. It
  follows that $p'_{i-1} + q'_{i-1} + r'_{i-1} = p_{i-1} + q_{i-1} + c_i \ge c_i +
  1 + c_i + 1 + c_i = 3c_i + 2$. Since $k_3 \ge r_{i-1} \ge c_i + 1 = \min(k_3,
  \floor{k/4} + h_i) + 1$, we can observe that $c_i$ must be equal to $\floor{k/4}
  + h_i$. In addition, by the induction hypothesis we have $s_{i-1} \ge r_{i-1} -
  h_{i-1} \ge c_i + 1 - 2h_i$. Merging those facts we can conclude that $p'_{i-1}
  + q'_{i-1} + r'_{i-1} + s'_{i-1} \ge 4c_i +3 - 2h_i = 4\floor{k/4} + 4h_i + 3 -
  2h_i \ge k$, so we are done in this case.
  
  In the case (2) we have $p'_{i-1} = p_{i-1}$ and $q'_{i-1} = b_i = \min(k_2, c_i 
  + h_i) < q_{i-1} \le p_{i-1}$. It follows that $p'_{i-1} + q'_{i-1} \ge 2b_i + 
  1$. Since $k_2 \ge q_{i-1} \ge b_i + 1 = \min(k_2, c_i + h_i) + 1$, we can 
  observe that $b_i$ must be equal to $c_i + h_i$. In addition, by the induction 
  hypothesis and Observation \ref{obs:abc}, we can bound $r'_{i-1}$ as $r'_{i-1} = 
  \min(c_i, r_{i-1}) \ge \min(b_i-h_i, q_{i-1} - h_{i-1}) \ge b_i + 1 - 2h_i$. 
  Therefore, $p'_{i-1} + q'_{i-1} + r'_{i-1} \ge 3b_i + 2 - 2h_i = 3c_i + 2 + h_i$. 
  Since $c_i$ is defined as $\min(k_3,\floor{k/4} + h_i)$, we have to consider two 
  subcases of the possible value of $c_i$. If $c_i = k_3 \le \floor{k/4} + h_i$, 
  then we have $k_2 \ge q_{i-1} \ge b_i + 1 = c_i + h_i + 1 \ge k_3 + 2$. By 
  Observation \ref{obs:mws}.(3), $k_3$ must be equal to $\floor{k/3}$, thus $3c_i + 
  2 = 3k_3 + 2 \ge k$, and we are done. Otherwise, we have $c_i =\floor{k/4} + h_i$ 
  and since $s'_{i-1} = s_{i-1} \ge r_{i-1} - h_{i-1} \ge b_i + 1 - 4h_i = c_i + 1 
  - 3h_i$, we can conclude that $p'_{i-1} + q'_{i-1} + r'_{i-1} + s'_{i-1} \ge 4c_i 
  + 3 - 2h_i = 4\floor{k/4} + 3 + 2h_i \ge k$.
  
  The last case $a_i < p_{i-1}$ can be proved be the similar arguments. Having Eqs.
  (\ref{eq:4}.\ref{eq:5},\ref{eq:6}), we can start proving the inequalities from the
  lemma. Observe that in Network 4 the values of vectors $\bar{w}^i$, $\bar{x}^i$,
  $\bar{y}^i$ and $\bar{z}^i$ are defined with the help of three types of sorters:
  $\sort^4$, $\sort^3$ and $\sort^2$. The smaller sorters are used, when the corresponding
  index is out of the range and an input item is not available. In the following analysis
  we would like to deal only with $\sort^4$ and in the case of smaller sorters we extend
  artificially their inputs and outputs with 1's at the left end and 0's at the right end.
  For example, in line 18 we have $\tuple{y^i_{j}, x^i_{j + h_i}} \gets
  \sort^2(y^{i-1}_{j}, x^{i-1}_{j + h_i})$ and we can analyse this operation as $\tuple{1,
  y^i_{j}, x^i_{j + h_i}, 0} \gets \sort^4(1, y^{i-1}_{j}, x^{i-1}_{j + h_i}, 0)$. The 0
  input corresponds to the element of $\bar{w}^{i-1}$ with index $j+2h_i$, where $j+2h_i >
  k_1 \ge a_i$, and the 1 input corresponds to the element of $\bar{z}^{i-1}$ with index
  $j-h_i$, where $j-h_i < 0$. A similar situation is in lines 10, 12, 16 and 22, where
  $\sort^2$ and $\sort^3$ are used. Therefore, in the following we assume that elements of
  input sequences $\bar{w}^{i-1}$, $\bar{x}^{i-1}$, $\bar{y}^{i-1}$ and $\bar{z}^{i-1}$
  with negative indices are equal to 1 and elements of the inputs with indices above
  $a_i$, $b_i$, $c_i$ and $\floor{k/4}$, respectively, are equal to 0. This assumption
  does not break the monotonicity of the sequences and we also have the property that
  $w^{i-1}_j = 1$ if and only if $j \le p'_{j-1}$ (and similar ones for $\bar{x}^{i-1}$
  and $q'_{i-1}$, and so on).
  
  It should be clear now that, under the assumption above, we have:
  \begin{subequations}
  \begin{align}
  w^i_j = & \min(w^{i-1}_j, x^{i-1}_{j-h_i}, y^{i-1}_{j-2h_i}, z^{i-1}_{j-3h_i}) 
  & \textrm{~for~} 1 \le j \le a_i,\\
  x^i_j = & \snd(w^{i-1}_{j+h_i}, x^{i-1}_{j}, y^{i-1}_{j-h_i}, z^{i-1}_{j-2h_i}) 
  & \textrm{~for~} 1 \le j \le b_i,\\
  y^i_j = & \trd(w^{i-1}_{j+2h_i}, x^{i-1}_{j+h_i}, y^{i-1}_{j}, z^{i-1}_{j-h_i}) 
  & \textrm{~for~} 1 \le j \le c_i, \displaybreak[0]\\
  z^i_j = & \max(w^{i-1}_{j+3h_i}, x^{i-1}_{j+2h_i}, y^{i-1}_{j+h_i}, z^{i-1}_{j})
  & \textrm{~for~} 1 \le j \le \floor{k/4},
  \end{align}
  \end{subequations}
  where $\snd$ and $\trd$ denote the second and the third smallest element of
  its input, respectively. Since the functions $\min$, $\snd$, $\trd$ and
  $\max$ are monotone and the input sequences are monotone. we can conclude
  that $w^i_j \ge w^i_{j+1}$, $x^i_j \ge x^i_{j+1}$, $y^i_j \ge y^i_{j+1}$ and
  $z^i_j \ge z^i_{j+1}$. Let $p_i$, $q_i$, $r_i$ and $s_i$ denote the numbers
  of 1's in them. Clearly, we have $p_i + q_i + r_i + s_i = p'_{i-1} +
  q'_{i-1} + r'_{i-1} + s'_{i-1}$, thus $p_i + q_i + r_i + s_i \ge \min(k,
  p_{i-1} + q_{i-1} + r_{i-1} + s_{i-1}) \ge \min(k, p_0 + q_0 + r_0 + s_0)$,
  by the induction hypothesis. Thus. we have proved monotonicity of
  $\bar{w}^i$, $\bar{x}^i$, $\bar{y}^i$ and $\bar{z}^i$ and that Eq.
  \eqref{eq:3} holds.
  
  To prove Eq. \eqref{eq:2} we would like to show that $p_i - h_i \le q_i$,
  $q_i - h_i \le r_i$, $r_i - h_i \le s_i$, that is, that the following
  equalities are true: $x^i_{p_i-h_i} = 1$, $y^i_{q_i-h_i} = 1$ and
  $z^i_{r_i-h_i} = 1$. The first equality follows from the fact that
  $w^i_{p_i} = 1$ and $w^i_{p_i} \le x^i_{p_i-h_i}$, because they are output
  in this order by a single sort operation. The other two equalities can be shown
  by the similar arguments.
  
  The last equation we have to prove is \eqref{eq:1}. By the induction hypothesis and our 
  assumption we have $w^{i-1}_j \ge x^{i-1}_j \ge y^{i-1}_j \ge z^{i-1}_j$, for any $j$. 
  We know also that the vectors are non-increasing. We use these facts to show that 
  $w^i_{q_i} = 1$, which is equivalent to $p_i \ge q_i$. Since $w^i_{q_i} = 
  \min(w^{i-1}_{q_i}, x^{i-1}_{q_i-h_i}, y^{i-1}_{q_i-2h_i}, z^{i-1}_{q_i-3h_i})$, we 
  have to prove that all the arguments of the $\min$ function are 1's. From the 
  definition of $q_i$ we have $1 = x^i_{q_i} = \snd(w^{i-1}_{q_i+h_i}, x^{i-1}_{q_i}, 
  y^{i-1}_{q_i-h_i}, z^{i-1}_{q_i-2h_i})$, thus the maximum of any pair of arguments of 
  $\snd$ must be 1. Now we can see that $w^{i-1}_{q_i} \ge \max(w^{i-1}_{q_i+h_i}, 
  x^{i-1}_{q_i}) \ge 1$, $x^{i-1}_{q_i-h_i} \ge \max(x^{i-1}_{q_i}, y^{i-1}_{q_i-h_i}) 
  \ge 1$, $y^{i-1}_{q_i-2h_i} \ge \max(y^{i-1}_{q_i-h_i}, z^{i-1}_{q_i-2h_i}) \ge 1$ and 
  $z^{i-1}_{q_i-3h_i} \ge \max(y^{i-1}_{q_i-h_i}, z^{i-1}_{q_i-2h_i}) \ge 1$. In the last 
  inequality we use the fact that for any $j$ it is true that $z^{i-1}_{j-2h_i} \ge 
  y^{i-1}_{j}$ (because $r'_{i-1} - 2h_i \le s'_{i-1}$, by Eq. \eqref{eq:5}). Thus, 
  $w^i_{q_i}$ must be 1 and we are done. The other two inequalities can be proved in the 
  similar way with the help of two additional relations: $x^{i-1}_{j-2h_i} \ge 
  w^{i-1}_{j}$ and $y^{i-1}_{j-2h_i} \ge x^{i-1}_{j}$, which follows from the induction 
  hypothesis.
\end{proof}

\begin{theorem}
  The output of Network \ref{net:4mw_merge} is top $k$ sorted.
\end{theorem}
\begin{proof}
  The $\zip$ operation outputs its input vectors in the row-major order. By Eq.
  \eqref{eq:3} of Lemma \ref{lma:4wsel}, we know that elements in the $out$
  sequence are dominated by the $k$ largest elements in the output vectors of the
  main loop. From the order diagram given in Fig. \ref{fig:4col} it follows that
  $\tuple{y_{i-1}, \max(z_{i-1}, w_{j+1}), w_j, x_j}$ dominates $\tuple{y_i, z_i,
  \min(z_{i-1}, w_{j+1}), x_{i+1}}$, thus, after the sorting operations in lines
  15--16, the values will appear in the row-major order. The two special cases,
  where the whole sequence of four elements is not available for the $\sort^4$
  operation, are covered by lines 17 and 18. If the vector $\check{x}^i$ does not
  have the element with index $k_4 + 1$, then just 3 elements are sorted in line
  17. If $k \bmod 4 = 3$ then the element $y_{k_4 + 1}$ is the last one in desired
  order, but the vector $\check{z}^i$ does not have the element with index $k_4 +
  1$, so the network must sort just 2 elements in line 18. Note that in the first
  $\floor{k/4}$ rows we have $4\floor{k/4}$ elements of the top $k$ ones, so
  the values of $k - 4\floor{k/4}$ leftmost elements in the row $k_4+1$ should be
  corrected.
\end{proof}

\subsection{4-Odd-Even Merging Networks}\label{sub:3}

In this subsection the second merging algorithm for four columns is presented (Network \ref{net:4oe_merge}).
The input to the procedure is $\tuple{pref(k_1,\bar{y}^1), \dots,
pref(k_m,\bar{y}^m)}$, where each $\bar{y}^i$ is the output of the recursive call in
Network \ref{net:oe_sel}. Again, the goal is to put the $k$ largest
elements in top rows. It is done by splitting each input sequence into two parts, one
containing elements of odd index, the other containing elements of even index. Odd sequences
and even sequences are then recursively merged (lines 7--8). The resulting sequences
are again split into odd and even parts, and combined using auxiliary Network \ref{net:4oe_combine}
(lines 9--10). Lastly, the combine operation is used again (line 11) to get the top
$k$ sorted sequence.

Our network is the generalization of the classic Odd-Even Merging Network by Batcher \cite{batcher}.
Network \ref{net:4oe_combine} is the standard combine operation of the Odd-Even Merging Network,
but with smaller number of comparators, since our network is the selection network.
For more information on these classes of networks and to see examples, see \cite{knuth}.

\begin{definition}
  A binary sequence $\bar{r}$ is top $k'$ {\em full} if it is of the form $1^{k'}\bar{r}'$,
  where $\bar{r}' \in \bool^{|\bar{r}|-k'}$.
\end{definition}

\begin{algorithm}
\caption{$oe\_combine^s_k$}\label{net:4oe_combine}
\begin{algorithmic}[1]
  \Require {A pair of sequences $\tuple{\bar{x}, \bar{y}}$,
  where $k \leq s = |\bar{x}| + |\bar{y}|$. }
  \State $\bar{z} \gets \zip(\bar{x}, \bar{y})$
  \ForAll {$i \in \{1,\dots,\floor{\min(k, |\bar{z}|-1)/2}\}$}
    \State $\sort^2(z_{2i}, z_{2i+1})$
  \EndFor
  \State \Return $\bar{z}$ 
\end{algorithmic}
\end{algorithm}

\begin{observation}\label{obs:oec2}
  Let $K,k,l \in \nat$ and $K \geq k+l$ and $0 \leq l - k \leq 2$. If $\bar{a}$ is top $l$
  full and $\bar{b}$ is top $k$ full
  then $oe\_combine^{|a| + |b|}_{K}(\bar{a},\bar{b})$ is top $l+k$ full.
\end{observation}

\begin{proof}
  If $l \in \{k,k+1\}$ or $b_{k+1}=1$, then $oe\_combine^{|a| + |b|}_{K}(\bar{a},\bar{b})$
  is top $l+k$ full due to the zip
  operation in line 1 (notice that additional sorting is not needed). If $l = k+2$ and $b_{k+1}=0$,
  then after the zip operation in line 1, the sequence $\bar{z}$ is of the form $1^{2k+1}01\bool^*$.
  Since $K \geq k+l=2k+2$, the elements $\tuple{z_{2k+2},z_{2k+3}} = \tuple{0,1}$ will be sorted in line 3.
  Therefore $\bar{z}$ will be of the form $1^{2k+2}\bool^*$, so $\bar{z}$ is top $l+k$ full.
\end{proof}

\begin{algorithm}
\caption{$4oe\_merge^s_k$}\label{net:4oe_merge}
\begin{algorithmic}[1]
  \Require {A tuple of top $k$ sorted sequences $\tuple{\bar{w}, \bar{x}, \bar{y}, \bar{z}}$, 
  where $s = |\bar{w}| + |\bar{x}| + |\bar{y}| + |\bar{z}|$,
  $|\bar{w}| \ge |\bar{x}| \ge |\bar{y}| \ge |\bar{z}|$ and $k \geq 1$.}
  \State $k_1 = |\bar{w}|$; ~~$k_2 = |\bar{x}|$; ~~$k_3 = |\bar{y}|$; ~~$k_4 = 
  |\bar{z}|$
  \If {$k_2 = 0$} \Return $\bar{w}$ \EndIf
  \If {$k_1 = 1$}
    \State $\bar{a} \gets \zip(\bar{w},\bar{x},\bar{y},\bar{z})$
    \State \Return $\sort^{|\bar{a}|}(\bar{a})$
  \EndIf
  \State $sa = \ceil{k_1/2} + \ceil{k_2/2} + \ceil{k_3/2} + \ceil{k_4/2}$;
       ~~$sb = \floor{k_1/2} + \floor{k_2/2} + \floor{k_3/2} + \floor{k_4/2}$
  \State $\bar{a} \gets 4oe\_merge^{sa}_{\min(sa,\floor{k/2}+2)}(
      \bar{w}_{\odd},\bar{x}_{\odd},\bar{y}_{\odd},\bar{z}_{\odd})$
      \Comment{Recursive calls.}
  \State $\bar{b} \gets 4oe\_merge^{sb}_{\min(sb,\floor{k/2})}(
      \bar{w}_{\even},\bar{x}_{\even},\bar{y}_{\even},\bar{z}_{\even})$
  \State $\bar{c} \gets oe\_combine_{\floor{k/2}+1}(\bar{a}_{\odd},\bar{b}_{\odd})$
  \State $\bar{d} \gets oe\_combine_{\floor{k/2}}(\bar{a}_{\even},\bar{b}_{\even})$
  \State \Return $oe\_combine_k(\bar{c},\bar{d})$ 
  \Comment{The output is top $k$ sorted.}
\end{algorithmic}
\end{algorithm}

\begin{observation}\label{obs:ones}
If $\bar{x}$ is a sorted sequence, then $|\bar{x}_{\odd}|_1 - |\bar{x}_{\even}|_1 \leq 1$.
\end{observation}

\begin{lemma}
  Let $k,k' \in \nat$ and $k' \le k$. If $\tuple{\bar{w}, \bar{x}, \bar{y}, \bar{z}}$ is a top $k'$ sorted tuple
  such that $k \leq s = |\bar{w}| + |\bar{x}| + |\bar{y}| + |\bar{z}|$
  and $|\bar{w}| \ge |\bar{x}| \ge |\bar{y}| \ge |\bar{z}|$, then 
  $4oe\_merge^s_k(\bar{w}, \bar{x}, \bar{y}, \bar{z})$ is top $k'$ sorted.
\end{lemma}

\begin{proof}
  By induction on $(k',s)$. First we introduce the schema of the proof, as it has many base cases:

  \begin{enumerate}
    \item When $k' \le s = 1$ or $k_1 = 1$.
    \item When $k'=s$ and $k_1 \geq 2$.
    \item When $k'=1$.
    \item When $k'=2$.
  \end{enumerate}

  \noindent After that we will prove the inductive step. Notice that we needed to include cases 2 and 4
  in the base step, because inductive step is only viable when $\floor{k'/2} + 2 \leq k'$, which is when $k' \geq 3$.
  
  {\em Ad. 1.} If $k' \le s = 1$, then from $k_1 \ge k_2 \ge k_3 \ge k_4$ we get that $k_1 = 1$ and
  $k_2 = 0$. Therefore $\bar{w}$ is the only non-empty input sequence and it consists of only one element,
  therefore it is sorted and the algorithm terminates on line 2. If $k_1=1$, then from $k_1 \ge k_2 \ge k_3 \ge k_4$
  we get that each input sequence has at most one element. Therefore using sorter of order at most 4 gives the desired
  result (lines 4--5).

  {\em Ad. 2.} If $k'=s$, then $k=s$. Assume that $k_1 \geq 2$. We prove this case by induction on $s$. Base case is covered by $(1)$. Take any $s>1$
  and assume that for any $s'<s$, the lemma holds. Consider the execution of $4oe\_merge^{s}_{s}(\bar{w}, \bar{x}, \bar{y}, \bar{z})$.
  Each input sequence is top $s$ sorted, therefore they are sorted. Which means that $\tuple{\bar{w}_{\odd},\bar{x}_{\odd},\bar{y}_{\odd},\bar{z}_{\odd}}$
  and $\tuple{\bar{w}_{\even},\bar{x}_{\even},\bar{y}_{\even},\bar{z}_{\even}}$ are also sorted. Since $\min(sa,\floor{s/2} + 2)=sa$, $\min(sb,\floor{s/2})=sb$,
  $sa < s$ (because $k_1 \geq 2$) and $sb<s$ we can
  use the induction hypothesis to get that $\bar{a}$ is top $sa$ sorted and $\bar{b}$ is top $sb$ sorted. Since $|\bar{a}|=sa$ and $|\bar{b}|=sb$ we get that
  $\bar{a}$ and $\bar{b}$ are sorted. Using Observation \ref{obs:ones} on the input sequences we also get that $|\bar{a}|_1 - |\bar{b}|_1 \leq 4$.
  Therefore sequences $\bar{a}_{\odd},\bar{b}_{\odd},\bar{a}_{\even},\bar{b}_{\even}$ are sorted and
  $|\bar{a}_{\odd}|_1 - |\bar{b}_{\odd}|_1 \leq 2$ and $|\bar{a}_{\even}|_1 - |\bar{b}_{\even}|_1 \leq 2$. It is easy to check that
  $|\bar{a}_{\odd}|_1 + |\bar{b}_{\odd}|_1 \leq \ceil{sa/2} + \ceil{sb/2} \leq \ceil{sa/2} + \floor{sb/2} + 1 = \ceil{s/4} + \floor{s/4} + 1 \leq \floor{s/2} + 1$, (for $s>1$).
  Similarly we get $|\bar{a}_{\even}|_1 + |\bar{b}_{\even}|_1 \leq \floor{sa/2} + \floor{sb/2} \leq \ceil{sa/2} + \floor{sb/2} = \ceil{s/4} + \floor{s/4} \leq \floor{s/2}$, (for $s>1$).
  Therefore from Observation \ref{obs:oec2} we get that $\bar{c}$ is top $|\bar{a}_{\odd}|_1 + |\bar{b}_{\odd}|_1$ full and $\bar{d}$ is top $|\bar{a}_{\even}|_1 + |\bar{b}_{\even}|_1$ full.
  Also since $\bar{a}_{\odd},\bar{b}_{\odd},\bar{a}_{\even},\bar{b}_{\even}$ were sorted, $\bar{c}$ and $\bar{d}$ are also sorted.
  Using Observation \ref{obs:ones} we get that $|\bar{c}|_1 - |\bar{d}|_1 = |\bar{a}_{\odd}|_1 + |\bar{b}_{\odd}|_1 - |\bar{a}_{\even}|_1 - |\bar{b}_{\even}|_1 \leq 2$.
  Also $|\bar{c}|_1 + |\bar{d}|_1 \leq \ceil{sa/2} + \ceil{sb/2} + \floor{sa/2} + \floor{sb/2} = sa + sb = s$. Therefore using Observation \ref{obs:oec2}
  we get that the sequence returned in line 11 is top $|\bar{c}|_1 + |\bar{d}|_1$ full and it is sorted, since $\bar{c}$ and $\bar{d}$ were sorted. This completes the
  inductive step and this case.

  {\em Ad. 3.} Let $k \geq k'=1$. We prove this case by induction on $s$.
  Assume $s > 1$ and $k_1 \geq 2$ (the other cases are covered by $(1)$). Consider two sub-cases:

  $(i)$. $|\bar{w}::\bar{x}::\bar{y}::\bar{z}|_1=0$. Therefore output of the algorithm is of the form $0^*$, so it is top $k'$ sorted.

  $(ii)$. $|\bar{w}::\bar{x}::\bar{y}::\bar{z}|_1>0$. Without loss of generality assume that $|\bar{w}|_1>0$,
  and because it is top $k'$ sorted, we get that $w_1=1$. Because of that $w^{\odd}_1=1$ and since $k_1\geq2$, we get that $sa < s$,
  so we can use induction hypothesis to get that $\bar{a}$ is top $\ceil{k'/2}=1$ sorted and $a_1=1$. From this we get that $a^{\odd}_1=1$,
  therefore $c_1=1$. From this we conclude that the output in line 11 is of the form $1\bool^*$, so it is top $k'$ sorted.
  
  {\em Ad. 4.} Let $k \geq k'=2$. Assume $s>2$ and $k_1 \geq 2$ (the other cases are covered by $(1)$ and $(2)$).
  From our assumptions $\ceil{k'/2} = \floor{k'/2}=1$. Therefore from case $(3)$ we get that $\bar{a}$ and $\bar{b}$
  are top 1 sorted. Notice that $\floor{k/2}+1 \geq 2$ and consider $\tuple{a_1,b_1},\tuple{a_2,b_2} \in \{\tuple{0,0},\tuple{1,0},\tuple{1,1}\}$.
  Using Observation \ref{obs:oec2}, in all possible cases $\bar{c}$ is top 2 sorted (regardless of the value of $a_3$) and $\bar{d}$ is top 1 sorted.
  We can now consider all possible cases for values of $c_1,c_2$ and $d_1$ and using Observation \ref{obs:oec2} again, conclude that the sequence returned
  in line 11 is top $2=k'$ sorted.
  
  We can finally prove the induction step. Assume $s \geq k \geq k' > 2$ and $k_1 \geq 2$. Consider three cases:
  
  \noindent (i) One of the input sequences (say $\bar{w}$) is top $k'$ full.
  Therefore $\bar{w}_{\odd}$ is top $\ceil{k'/2}$ full and $\bar{w}_{\even}$ is top $\floor{k'/2}$ full.
  Also since $s \geq k' >2$ and $k_1 \geq 2$ we get that $sa < s$, $sb < s$ and $\floor{k'/2}+2 \leq k'$. Therefore
  we can use induction hypothesis to get that $\bar{a}$ is top $\ceil{k'/2}$ full and $\bar{b}$ is top $\floor{k'/2}$ full.
  Therefore $\bar{a}_{\odd}$ is top $\ceil{\ceil{k'/2}/2}$ full,$\bar{a}_{\even}$ is top $\floor{\ceil{k'/2}/2}$
  full, $\bar{b}_{\odd}$ is top $\ceil{\floor{k'/2}/2}$ full and $\bar{b}_{\even}$ is top $\floor{\floor{k'/2}/2}$ full.
  Let $q = \ceil{\floor{k'/2}/2} + \ceil{\floor{k'/2}/2}$ and $r = \floor{\ceil{k'/2}/2} + \floor{\floor{k'/2}/2}$. Notice that
  $q \leq \floor{k'/2} + 1$ and $r \leq \floor{k'/2}$ (for $k' \geq 2$).
  Furthermore, $0 \leq \ceil{\ceil{k'/2}/2} - \ceil{\floor{k'/2}/2} \leq 1$ and $0 \leq \floor{\ceil{k'/2}/2} - \floor{\floor{k'/2}/2} \leq 1$.
  Hence we can use Observation \ref{obs:oec2} to get that $\bar{c}$ is top $q$ full
  and $\bar{d}$ is top $r$ full. It is easy to check that $0 \leq q - r \leq 2$ and $q+r=k'$,
  therefore we can use Observation \ref{obs:oec2} again and we can conclude that
  the last combine operation (line 11) will return top $k'$ sorted sequence, therefore this case is done.

  \noindent (ii) All input sequences are of the form $1^j0^*$, where $j<k'$, and $I = |\bar{w}::\bar{x}::\bar{y}::\bar{z}|_1 \leq k'$.
  Since $s \geq k' >2$ and $k_1 \geq 2$ we get that $sa < s$, $sb < s$ and $\floor{k'/2}+2 \leq k'$. Therefore
  using induction hypothesis we get that $\bar{a}$ is of the form $1^q0^*$ and $\bar{b}$ is of the form
  $1^r0^*$, where $q + r = I$. Furthermore, applying Observation \ref{obs:ones} to all input sequences we get that $0 \leq q - r \leq 4$.
  Which means that $\bar{a}_{\odd},\bar{a}_{\even},\bar{b}_{\odd},\bar{b}_{\even}$ are all sorted and of the form
  $1^{\ceil{q/2}}0^*$, $1^{\floor{q/2}}0^*$, $1^{\ceil{r/2}}0^*$, $1^{\floor{r/2}}0^*$ respectively, also
  $0 \leq |\bar{a}_{\odd}|_1 - |\bar{b}_{\odd}|_1 = \ceil{q/2} - \ceil{r/2} \leq 2$ and $0 \leq |\bar{a}_{\even}|_1 - |\bar{b}_{\even}|_1 = \floor{q/2} - \floor{r/2} \leq 2$.
  Using Observation \ref{obs:oec2} we get that $\bar{c}$ is top $\ceil{q/2} + \ceil{r/2}$ full and $\bar{d}$
  is top $\floor{q/2} + \floor{r/2}$ full.
  Moreover, $0 \leq |\bar{c}|_1 - |\bar{d}|_1 = \ceil{q/2} + \ceil{r/2} - \floor{q/2} - \floor{r/2} \leq 2$, therefore using
  Observation \ref{obs:oec2} again we conclude
  that the returned sequence (in line 11) is top $q+r$ full, and since $q+r = I < k'$,
  then the returned sequence is of the form $1^{I}0^*$, so it is also top $k'$ sorted.

  \noindent (iii) All input sequences are of the form $1^j0^*$,
  where $j<k'$, and $I = |\bar{w}::\bar{x}::\bar{y}::\bar{z}|_1 > k'$.
  Since $s \geq k' >2$ and $k_1 \geq 2$ we get that $sa < s$, $sb < s$ and $\floor{k'/2}+2 \leq k'$.
  Therefore we can use induction hypothesis to get that $\bar{a}$ is top $\floor{k'/2} + 2$
  sorted and $\bar{b}$ is top $\floor{k'/2}$ sorted.
  Applying Observation \ref{obs:ones} to all input sequences we get that $0 \leq |\bar{a}|_1 - |\bar{b}|_1 \leq 4$,
  in particular, $|\bar{a}|_1 \geq |\bar{b}|_1$.
  Let $q = |a_{1\dots\floor{k'/2} + 2}|_1$ and $r = |b_{1\dots\floor{k'/2}}|_1$.
  Furthermore, either $\bar{a}$ is top $\floor{k'/2} + 2$ full or $\bar{b}$ is top $\floor{k'/2}$ full.
  Otherwise $a_{\floor{k'/2} + 2}=0$ and $b_{\floor{k'/2}}=0$, which means that for each $j \geq \floor{k'/2} + 2$
  and $j' \geq \floor{k'/2}$, $a_j=0$ and $b_{j'}=0$. This implies that
  $I \leq q + r \leq (\floor{k'/2} + 1) + (\floor{k'/2} - 1) \leq k'$, a contradiction. Consider two cases:
  (A) $\bar{b}$ is top $\floor{k'/2}$ full. Since $|\bar{a}|_1 \geq |\bar{b}|_1$, then $\bar{a}$ is also top $\floor{k'/2}$ full.
  Furthermore, since $I > k'$ it is true that $\bar{a}$ is also top $\floor{k'/2} + 1$.
  This means that $q+r \geq k'$. For case (B) assume that $\bar{a}$ is top $\floor{k'/2} + 2$ full.
  If $a_{\floor{k'/2} + 3}=0$, then since $I>k'$ we get that $\bar{b}$ is top $\floor{k'/2}-1$ full. This implies that $q+r \geq k'$.
  On the othe hand, if $a_{\floor{k'/2} + 3}=1$, then since $|\bar{a}|_1 - |\bar{b}|_1 \leq 4$ we get that $\bar{b}$ is top $\floor{k'/2}-1$ full,
  which also implies that $q+r \geq k'$. Therefore in all the cases $q+r \geq k'$.
  Notice that $0 \leq \ceil{q/2} - \ceil{r/2} \leq 2$ and $0 \leq \floor{q/2} - \floor{r/2} \leq 2$.
  Using Observation \ref{obs:oec2} we get that $\bar{c}$ is top $\ceil{q/2} + \ceil{r/2}$ full and $\bar{d}$
  is top $\floor{q/2} + \floor{r/2}$ full. Moreover,
  $0 \leq \ceil{q/2} + \ceil{r/2} - \floor{q/2} - \floor{r/2} \leq 2$, therefore using
  Observation \ref{obs:oec2} again we conclude
  that the returned sequence (in line 11) is top $q+r$ full, and since $q+r \geq k'$,
  then the returned sequence is top $k'$ sorted.
\end{proof}

\section{Comparison of Selection networks} \label{sec:comp}

In this section we would like to estimate and compare the number of variables and clauses in encodings
based on our algorithms to other encoding based
on odd-even selection and pairwise selection proposed by Codish and Zazon-Ivry \cite{card2}.
We start with counting how many variables and clauses are needed in order to merge 4 sorted sequences returned
by recursive calls of 4-Wise Selection Network, Pairwise Cardinality Network,
4-Odd-Even Selection Network and Odd-Even Selection Network. Then, based on those values we
prove that overall number of variables and clauses are almost
always smaller when using 4-column encodings rather than 2-column encodings.
In the next section we show that the new encodings are not just smaller, but
also have better solving times in many benchmark instances.

To simplify the presentation we assume that $k \leq n/4$ and both $k$ and $n$
are the powers of $4$. We also omit the ceiling and floor
function in the calculations, when it is convenient for us.

\begin{definition}
  Let $n,k \in \nat$. For given (selection) network $f^n_k$ let $V(f^n_k)$ and $C(f^n_k)$ denote the number
  of variables and clauses used in the standard CNF encoding of $f^n_k$.
\end{definition}

We remind the reader that a single 2-comparator uses 2 auxiliary variables and 3 clauses.
In case of 3-comparator the numbers are 3 and 7, and
for 4-comparator it's 4 and 15.

We begin by briefly introducing the Pairwise Cardinality Networks \cite{card2} that can
be summed up by the following three-step recursive procedure (omitting the base case):

\begin{enumerate}
\item Split the input $\bar{x} \in \bool^n$ into two sequences $\bar{x}^1$ and $\bar{x}^2$
  of equal size and order pairs, such that for each $1 \leq i \leq n/2$, $\bar{x}^1_i \geq \bar{x}^2_i$.
\item Recursively select top $k$ sorted elements from $\bar{x}^1$ and top $k/2$ sorted
  elements from $\bar{x}^2$.
\item Merge the outputs of the previous step using a Pairwise Merging Network of order $(2k,k)$ and
  output the top $k$ elements.
\end{enumerate}

The details of the above construction can be found in \cite{card2}.
If we treat the splitting step as a network $pw\_split^n$,
the merging step as a network $pw\_merge^{2k}_k$, then the number of 2-comparators used in the
Pairwise Cardinality Network of order $(n,k)$ can be written as:

\begin{equation}
|pw\_sel^n_k| = \left\{ 
    \begin{array}{l l}
      |pw\_sel^{n/2}_k| + |pw\_sel^{n/2}_{k/2}|+ &  \\ 
      + |pw\_split^n| + |pw\_merge^{2k}_k| & \quad \text{if $k<n$}\\
      |oe\_sort^k| & \quad \text{if $k=n$} \\
      |max^n| & \quad \text{if $k=1$} \\
    \end{array} \right.
    \label{eq:pw}
\end{equation}

One can check that
Step 1 requires $|pw\_split^n| = n/2$ 2-comparators and Step 3 requires $|pw\_merge^{2k}_k|=k\log k - k + 1$
2-comparators. The schema of this network is presented in Figure \ref{fig:pw_sel}. We want to count
the number of variables and clauses used in merging 4 outputs of the recursive steps, therefore we expand the recursion
by one level (see Figure \ref{fig:pw_sel2}).

 \begin{figure}
  \centering
      \subfloat[one step\label{fig:pw_sel1}]{%
      \includegraphics[width=0.48\textwidth]{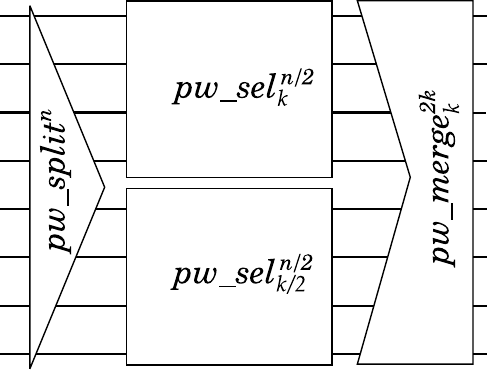}
      }
      ~
      \subfloat[two steps\label{fig:pw_sel2}]{%
      \includegraphics[width=0.48\textwidth]{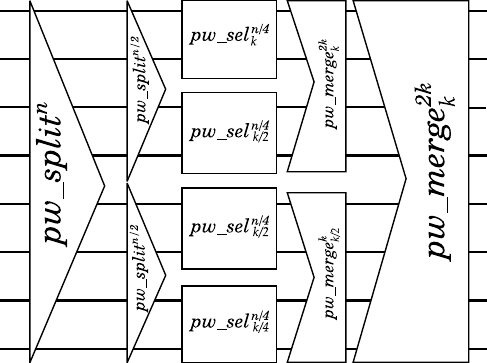}
      }
      \caption{Pairwise Cardinality Network}
    \label{fig:pw_sel}
  \end{figure}

 \begin{lemma}
   Let $k \in \nat$. Then $2V(pw\_merge^{2k}_k) + V(pw\_merge^{k}_{k/2}) = 5k\log k -6k +6$
   and $2C(pw\_merge^{2k}_k) + C(pw\_merge^{k}_{k/2}) =\frac{15}{2}k\log k - 9k + 9$.
 \end{lemma}

 \begin{proof}
   The number of 2-comparator used is:
   
   \[
     2|pw\_merge^{2k}_k| + |pw\_merge^{k}_{k/2}| = \frac{5}{2}k \log k - 3k + 3
   \]

   Elementary calculation gives the desired result.
 \end{proof}

 We now count the number of variables and clauses for the 4-Wise Selection Network, again,
 disregarding the recursive steps.

 \begin{lemma}
   Let $k \in \nat$. Then $V(4w\_merge^{4k}_k) = k\log k + \frac{7}{6}k - 5$,
   $C(4w\_merge^{4k}_k) = \frac{15}{4}k \log k - \frac{33}{24}k - 10$.
 \end{lemma}

 \begin{proof}
   We separately count the number of 2, 3 and 4-comparators used in the merger (Network \ref{net:4mw_merge}).
   By the assumption that $k \leq n/4$ we get $|\bar{w}|=k$, $|\bar{x}|=k/2$, $|\bar{y}|=k/3$,
   $|\bar{z}|=k/4$ and $h_1=k/2$. We will consider iterations 1, 2 and 3 separately and then provide
   the formulas for the number of comparators for iterations 4 and beyond. Results are summarized in
   Table \ref{tbl:4w_comp}.

   In the first iteration $h_1 = k/2$, which means that sets of $j$-values
   in the first two inner loops (lines 7--10 and 11--13) are empty. On the other hand, $1 \leq j \leq \min(k_1 - h_1, k_2, h_1) = k/2$
   (line 14), therefore $k/2$ 2-comparators are used. In fact it is
   true for every iteration $i$ that $\min(k_1 - h_i, k_2, h_i) = h_i = k/2^i$. We note this fact in column 3 of Table \ref{tbl:4w_comp}
   (the first term of each expression). For the next iterations we only need to consider the
   first and second inner loops of the algorithm.

   In the second iteration $h_2 = k/4$, therefore in the first inner loop only the condition in line 10 will hold, and only
   when $j \leq k/3 - k/4 = k/12$, hence $k/12$ 2-comparators are used. In the second inner loop the $j$-values satisfy condition
   $1 \leq j \leq \min(k_2 - h_2,k_3,h_2)=k/4$. Therefore the condition in line 12
   will hold for each $j \leq k/4$, hence $k/4$ 3-comparators
   are used. In fact it is true for every iteration $i \geq 2$ that $\min(k_2 - h_i,k_3,h_i)=k/2^i$,
   so the condition in line 12 will be true for $1 \leq j \leq k/2^i$.
   We note this fact in column 4 of Table \ref{tbl:4w_comp}. For the next iterations we only need to consider the
   first inner loop of the algorithm.

   In the third iteration $h_3 = k/8$, therefore $j$-values in the first inner loop satisfy the condition
   $1 \leq j \leq \min(k_3 - h_3,k_4) = 5k/24$ and
   the condition in line 8 will hold for each $j \leq 5k/24$, hence
   $5k/24$ 4-comparators are used.
   
   From the forth iteration $h_i \leq k/16$, therefore for every $1 \leq j \leq \min(k_3 - h_i,k_4) = k/4$ the condition in line 8
   holds. Therefore $k/4$ 4-comparators are used.

   What's left is to sum 2,3 and 4-comparators throughout $\log k$ iterations of the algorithm.
   The results are are presented in Table \ref{tbl:4w_comp}.
   Elementary calculation gives the desired result.
   
   \begin{table}[t] \setlength{\tabcolsep}{4pt}
     \centering
     \bgroup
     \def\arraystretch{1.5}
     \begin{tabular}{ c | c || c | c | c | } \hline
       \multicolumn{1}{|c|}{$i$} & $h_i$ & \#2-comparators & \#3-comparators & \#4-comparators \\ \hline
       \multicolumn{1}{|c|}{$1$} & $\frac{k}{2}$ & $\frac{k}{2}$ & $0$ & $0$ \\ \hline
       \multicolumn{1}{|c|}{$2$} & $\frac{k}{4}$ & $\frac{k}{4} + \frac{k}{12}$ & $\frac{k}{4}$ & $0$ \\ \hline
       \multicolumn{1}{|c|}{$3$} & $\frac{k}{8}$ & $\frac{k}{8}$ & $\frac{k}{8}$ & $\frac{5k}{24}$ \\ \hline
       \multicolumn{1}{|c|}{$\geq4$} & $\frac{k}{2^i}$ & $\frac{k}{2^i}$ & $\frac{k}{2^i}$ & $\frac{k}{4}$ \\ \hline
       & Sum & $\frac{13}{12}k - 1$ & $\frac{k}{2} - 1$ & $\frac{1}{4} k \log k - \frac{13k}{24}$ \\ \cline{2-5} 
     \end{tabular}
     \def\abovecaptionskip{10pt}
     \caption{Number of comparators used in different iterations of Network \ref{net:4mw_merge}}\label{tbl:4w_comp}
     \egroup
   \end{table}
\end{proof}

 Let $V_{PCN}=2V(pw\_merge^{2k}_k) + V(pw\_merge^{k}_{k/2})$, $C_{PCN}=2C(pw\_merge^{2k}_k) + C(pw\_merge^{k}_{k/2})$,
 $V_{4W}=V(4w\_merge^{4k}_k)$ and $C_{4W}=C(4w\_merge^{4k}_k)$. Then the following corollary shows that
 using 4-column pairwise merging networks produces smaller encodings than their
 2-column counterpart.
 
 \begin{corollary}\label{crly:pw}
   Let $k \in \nat$ such that $k \geq 4$. Then $V_{4W} < V_{PCN}$ and $C_{4W} < C_{PCN}$.
 \end{corollary}

 \noindent For the following theorem, note that $V(4w\_split^n)=V(pw\_split^n)=n$.
 
 \begin{theorem}
    Let $n,k \in \nat$ such that $1\leq k \leq n/4$ and $n$ and $k$ are both powers of 4. Then
    $V(4w\_sel^n_k) \leq V(pw\_sel^n_k)$.
  \end{theorem}

  \begin{proof}
    Proof is by induction. For the base case, consider $1=k<n$. It follows that
    $V(4w\_sel^n_k)=V(pw\_sel^n_k)=V(max^n)$. For the induction step assume that for each
    $(n',k') \prec (n,k)$ (in lexicographical order), where $k \geq 4$, the inequality holds, we get:

    \[
    \begin{tabular}{l r}
      \multicolumn{2}{l}{$V(4w\_sel^n_k)= V(4w\_split^n) + \sum_{1 \leq i \leq 4} V(4w\_sel^{n/4}_{k/i}) + V(4w\_merge^{4k}_k)$} \\
      & \small{\bf (by the construction of $4w\_sel$)} \\
      \multicolumn{2}{l}{$\leq V(4w\_split^n) + \sum_{1 \leq i \leq 4} V(pw\_sel^{n/4}_{k/i}) + V(4w\_merge^{4k}_k)$} \\
      & \small{\bf (by the induction hypothesis)} \\
      \multicolumn{2}{l}{$\leq V(pw\_split^n) + 2V(pw\_split^{n/2}) + \sum_{1 \leq i \leq 4} V(pw\_sel^{n/4}_{k/i})$} \\
      \multicolumn{2}{l}{$\quad + 2V(pw\_merge^{2k}_k) + V(pw\_merge^{k}_{k/2})$} \\
      & \small{\bf (by Corollary \ref{crly:pw} and because $V(4w\_split^n) < V(pw\_split^n) + 2V(pw\_split^{n/2})$)} \\
      \multicolumn{2}{l}{$\leq V(pw\_split^n) + 2V(pw\_split^{n/2}) + V(pw\_sel^{n/4}_{k})+ 2V(pw\_sel^{n/4}_{k/2})$} \\
      \multicolumn{2}{l}{$\quad + V(pw\_sel^{n/4}_{k/4}) + 2V(pw\_merge^{2k}_k) + V(pw\_merge^{k}_{k/2})$} \\
      & \small{\bf (because $V(pw\_sel^{n/4}_{k/3}) \leq V(pw\_sel^{n/4}_{k/2})$)} \\
      \multicolumn{2}{l}{$= V(pw\_sel^n_k)$} \\
      & \small{\bf (by the construction of $pw\_sel$)} \\
    \end{tabular}
    \]
  \end{proof}
 
  Now we will count how many variables and clauses are needed in order to merge 4 sorted sequences
  returned by recursive calls of Odd-Even Selection Network and 4-Odd-Even Selection Network, respectively.
  Two-column selection network using odd-even approach is presented in \cite{card2}.
  We briefly introduce this network with the following three-step recursive procedure (omitting the base case):

  \begin{enumerate}
    \item Split the input $\bar{x} \in \bool^n$ into two sequences $\bar{x}^1 = \bar{x}_{\odd}$
      and $\bar{x}^2 = \bar{x}_{\even}$.
    \item Recursively select top $k$ sorted elements from $\bar{x}^1$ and top $k$ sorted
      elements from $\bar{x}^2$.
    \item Merge the outputs of the previous step using a Odd-Even Merging Network of order $(2k,k)$ and
      output the top $k$ elements.
  \end{enumerate}

  If we treat the merging step as a network $oe\_merge^{2k}_k$, then the number of 2-comparators used in the
  Odd-Even Selection Network of order $(n,k)$ can be written as:

  \begin{equation}
    |2oe\_sel^n_k| = \left\{ 
    \begin{array}{l l}
      2|2oe\_sel^{n/2}_k| + |2oe\_merge^{2k}_k| & \quad \text{if $k<n$}\\
      |oe\_sort^k| & \quad \text{if $k=n$} \\
      |max^n| & \quad \text{if $k=1$} \\
    \end{array} \right.
    \label{eq:oe}
  \end{equation}
  
  One can check that Step 3 requires $|2oe\_merge^{2k}_k|=k\log k + 1$ 2-comparators.
  Overall number of 2-comparators was already calculated in \cite{card2}.

  \begin{lemma}[\cite{card2}]\label{lma:2oe_comp}
    For $k \leq n$ (both powers of two) $2oe\_sel^n_k$
    uses $\frac{1}{4}n\log^2 k + \frac{3}{4}n \log k - k \log k$
    2-comparators.
  \end{lemma}

  The schema of this network is presented in Figure \ref{fig:oe_sel}.
  Similarly to the previous case, in order to count the number of
  comparators used in merging 4 sorted sequences we need
  to expand the recursive step by one level (see Figure \ref{fig:oe_sel2}).

  \begin{figure}
    \centering
    \subfloat[one step\label{fig:oe_sel1}]{%
      \includegraphics[width=0.48\textwidth]{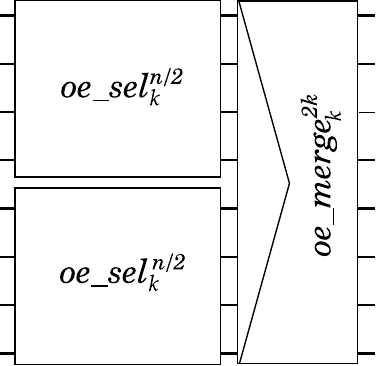}
    }
    ~
    \subfloat[two steps\label{fig:oe_sel2}]{%
      \includegraphics[width=0.48\textwidth]{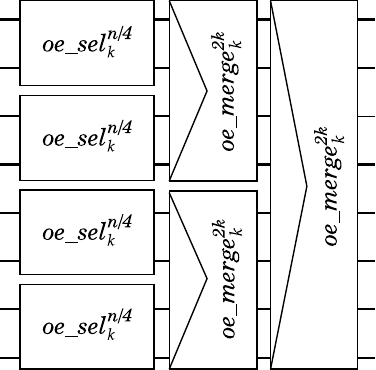}
    }
    \caption{Pairwise Cardinality Network}
    \label{fig:oe_sel}
  \end{figure}

  \begin{lemma}\label{lma:2oe_vars_cls}
   $3V(2oe\_merge^{2k}_k) = 6 k \log k + 6$, $3C(2oe\_merge^{2k}_k) = 9 k \log k + 9$.
  \end{lemma}

  \begin{proof}
    The number of 2-comparator used is:
   
   \[
     3|2oe\_merge^{2k}_k| = 3 k \log k + 3
   \]

   Elementary calculation gives the desired result.
  \end{proof}
  
  Now we do the same counting for our 4-Odd-Even Selection Network based on
  Networks \ref{net:oe_sel} and \ref{net:4oe_merge}.
  
  \begin{lemma}\label{lma:4oe_vars_cls}
    Let $k \in \nat$. Then $V(4oe\_merge^{4k}_k) \leq 2k \log k + 16k - 4 \log k - 6$
    and $C(4oe\_merge^{4k}_k) \leq 3k \log k + 33k - 6 \log k - 9$.
  \end{lemma}

  \begin{proof}
    We separately count the number of 2 and 4-comparators used.
    We count the comparators in the merger (Network \ref{net:4oe_merge}).
    By the assumption that $k \leq n/4$ and that $k$ is the power of 4 we get
    that $|\bar{w}|=|\bar{x}|=|\bar{y}|=|\bar{z}|=k$. This implies that $s=4k$ and
    $sa=sb=2k$.

    In the algorithm 4-comparators are only used in the base case (line 3--5). Furthermore
    number of 4-comparators is only dependent on the variable $s$. The solution to the following recurrence
    gives the sought number: $\{A(4) = 1; A(s) = 2A(s/2), \text{for } s>4\}$, which is equal to $s/4 = k$.

    We now count the number of 2-comparators used. The combine operations in lines 9--11
    uses $(2k/2 + 1 + k)/2 \leq k + 1$ 2-comparators. The total number of 2-comparators is then
    a solution to the following recurrence:

    \[
      B(s,k) =
        \begin{cases}
          0                     & \quad \text{if } s \leq 4 \\
          B(s/2,k/2+2)+B(s/2,k/2)+k+1  & \quad \text{otherwise } \\
        \end{cases}
    \]

    \noindent We claim that $B(s,k) \leq k \log \frac{s}{2} - k - 2 \log s + \frac{3}{2}s + 1$.
    This can be easily verified using induction. Applying $s=4k$ gives us the upper bound
    on number of 2-comparators used, which is $k \log k + 6k - 2\log k - 3$.

    Elementary calculation gives the desired result.
  \end{proof}

  Let $V_{2OE}=3V(oe\_merge^{2k}_k)$, $C_{2OE}=3C(oe\_merge^{2k}_k)$,
  $V_{4OE}=V(4oe\_merge^{4k}_k)$ and $C_{4OE}=C(4oe\_merge^{4k}_k)$.
  
  \begin{corollary}\label{crly:oe}
    Let $k \in \nat$. If $k \geq 16$, then $V_{4OE} \leq V_{2OE}$. If $k \geq 64$, then $C_{4OE} \leq C_{2OE}$
  \end{corollary}

  The corollary above does not cover all values of $k$. For the sake of the next theorem we set $k=4$ and
  compare values $V(4oe\_sel^n_4)$ and $V(2oe\_sel^n_4)$ in the following lemma. We do similar calculations
  for $C(4oe\_sel^n_k)$ and $C(2oe\_sel^n_k)$, for $k \in \{4,16\}$.
  
  \begin{lemma}\label{lma:k=4}
    Let $n \in \nat$ such that $n \geq 16$ and $n$ is a power of 4.
    Then $V(4oe\_sel^n_4) \leq V(2oe\_sel^n_4)$.
  \end{lemma}

  \begin{proof}
    Applying $k=4$ to the formula in Lemma \ref{lma:2oe_comp} gives $\frac{5}{2}n-8$
    2-comparators and therefore $V(2oe\_sel^{n}_4)=5n - 16$.

    On the other hand, we calculate the value $B(s,k)$ from Lemma \ref{lma:4oe_vars_cls} more accurately
    when $k=4$ to get accurate result for $V(4oe\_merge^{s}_4)$ and, in consequence, for $V(4oe\_sel^n_4)$.
    First we rewrite the formula so it represents the number of variables used in $4oe\_merge^n_k$:

    \[
      B^*(s,k) =
        \begin{cases}
          4                                                           & \quad \text{if } s \leq 4 \\
          B^*(s/2,\floor{k/2}+2)+B^*(s/2,\floor{k/2})+4\floor{k/2}+2  & \quad \text{otherwise } \\
        \end{cases}
    \]

    \noindent The change is straightforward. Since $s=4k$, we can calculate by hand that $B^*(16,4)=42$.
    The sought value is given by:

    \[
      V(4oe\_sel^n_4) =
        \begin{cases}
          V(4oe\_merge^{16}_4) + 16                   & \quad \text{if } n = 16 \\
          4V(4oe\_sel^{n/4}_4) + V(4oe\_merge^{16}_4)  & \quad \text{otherwise } \\
        \end{cases}
    \]

    \noindent In the first case we use one merger of order $(16,4)$ and 4 4-comparators.
    The second case is straightforward. Since $V(4oe\_merge^{16}_4)=B^*(16,4)$, so the recursive formula
    above can be reduced into:

    \[
    f(n) =
    \begin{cases}
          58                   & \quad \text{if } n = 16 \\
          4f(n/4) + 42         & \quad \text{otherwise } \\
        \end{cases}
    \]

    \noindent for which the solution is $f(n)=\frac{9}{2}n-14$. Obviously
    $\frac{9}{2}n-14 \leq  5n - 16$, for $n \geq 4$, so we are done.
  \end{proof}

  \begin{theorem}
    Let $n,k \in \nat$ such that $1\leq k \leq n/4$ and $n$ and $k$ are both powers of 4. Then
    $V(4oe\_sel^n_k) \leq V(2oe\_sel^n_k)$.
  \end{theorem}

  \begin{proof}
    Proof is by induction. For the base case, consider $1=k<n$. It follows that
    $V(2oe\_sel^n_k)=V(4oe\_sel^n_k)=V(max^n)$. The case $n>k=4$ is handled by Lemma \ref{lma:k=4}.
    For the induction step assume that for each
    $(n',k') \prec (n,k)$ (in lexicographical order), where $k \geq 16$, the inequality holds, we get:

    \[
    \begin{tabular}{l r}
      \multicolumn{2}{l}{$V(4oe\_sel^n_k)=4V(4oe\_sel^{n/4}_{k}) + V(4oe\_merge^{4k}_k)$} \\
      & \small{\bf (by the construction of $4oe\_sel$)} \\
      \multicolumn{2}{l}{$\leq 4V(2oe\_sel^{n/4}_{k}) + V(4oe\_merge^{4k}_k)$} \\
      & \small{\bf (by the induction hypothesis)} \\
      \multicolumn{2}{l}{$\leq 4V(2oe\_sel^{n/4}_{k}) + 3V(2oe\_merge^{2k}_k)$} \\
      & \small{\bf (by Corollary \ref{crly:oe})} \\
      \multicolumn{2}{l}{$= V(2oe\_sel^n_k)$} \\
      & \small{\bf (by the construction of $2oe\_sel$)} \\
    \end{tabular}
    \]
  \end{proof}
  
\section{Experimental Evaluation} \label{sec:exp}
  
As it was observed in \cite{abio}, having a smaller encoding in terms of number of
variables or clauses is not always beneficial in practice, as it should also be
accompanied with a reduction of SAT-solver runtime. In this section we assess how our
encodings based on the new family of selection networks affect the performance of a
SAT-solver.

Our algorithms that encode CNF instances with cardinality constraints into CNFs were
implemented as an extension of \textsc{MiniCard} ver. 1.1, created by Mark Liffiton and
Jordyn Maglalang\footnote{See https://github.com/liffiton/minicard}. \textsc{MiniCard}
uses three types of solvers:

\begin{itemize}
  \item {\em minicard} - the core \textsc{MiniCard} solver with native AtMost constraints,
  \item {\em minicard\_encodings} - a cardinality solver using CNF encodings for AtMost constraints,
  \item {\em minicard\_simp\_encodings} - the above solver with simplification / preprocessing.
\end{itemize}

The main program in {\em minicard\_encodings} has an option to generate a CNF formula, given a CNFP
instance (CNF with the set of cardinality constraints)
and to select a type of encoding applied to cardinality constraints. Program run with this option outputs a CNF
instance that consists of collection of the original clauses with the conjunction of CNFs generated by given method
for each cardinality constraint. No additional preprocessing and/or simplifications are made.
Authors of {\em minicard\_encodings} have implemented six methods to encode cardinality constraints and arranged them in one library called {\em Encodings.h}. 
Our modification of \textsc{MiniCard} is that we added implementations of encodings presented in this paper and put them in
the library {\em Encodings\_MW.h}. Then, for each CNFP instance and each encoding method, we used \textsc{MiniCard} to generate CNF instances.
After preparing, the benchmarks were run on a different SAT-solvers. Our extension of \textsc{MiniCard}, which we call \textsc{KP-MiniCard},
is available online\footnote{See https://github.com/karpiu/kp-minicard}.

We use two SAT-solvers in this evaluation. The first one is the well known \textsc{MiniSat} (ver. 2.2.0)
created by Niklas E\'en and Niklas S\"orensson\footnote{See https://github.com/niklasso/minisat}.
The second SAT-solver is the \textsc{GlueMiniSat} (ver. 2.2.8), by Hidetomo Nabeshima, Koji Iwanuma and
Katsumi Inoue\footnote{See http://glueminisat.nabelab.org/}.
All experiments were carried out on the machines with Intel(R) Core(TM) i7-2600 CPU @ 3.40GHz.

Detailed results are available online\footnote{See http://www.ii.uni.wroc.pl/\%7ekarp/sat/2017.html }.
We publish spreadsheets showing running time for each instance,
speed-up/slow-down tables for our encodings, number of time-outs met and total running time.

\vspace{0.3cm}

\noindent {\bf Encodings:} From this paper we choose two multi-column selection networks
for evaluation:

\begin{itemize}
  \item {\bf 4WISE} - 4-column selection network based on pairwise approach (Networks \ref{net:mw_sel} and \ref{net:4mw_merge}),
  \item {\bf 4OE} - 4-column odd-even selection network based on Networks \ref{net:oe_sel} and \ref{net:4oe_merge}.
\end{itemize}

We compare our encodings to some others found in the literature. We consider
the Pairwise Cardinality Networks \cite{card2} to rival our 4-column
pairwise selection network. We also consider a solver called \textsc{MiniSat+}\footnote{See https://github.com/niklasso/minisatp}
which implements techniques to encode Pseudo-Boolean constraints to propositional formulas \cite{card3}.
Since cardinality constraints are a subclass of Pseudo-Boolean constraints, we can
measure how well the encodings used in \textsc{MiniSat+} perform, compared with our methods.
The solver chooses between three techniques to generate SAT encodings for Pseudo-Boolean constraints.
These convert the constraint to: a BDD structure, a network of binary adders, a network of sorters.
The network of adders is the most concise encoding, but it has poor propagation properties and often
leads to longer computations than the BDD based encoding. The network of sorters is the implementation
of classic odd-even (2-column) sorting network by Batcher \cite{batcher}. Calling the solver we can
choose the encoding with one of the parameters: {\em -ca}, {\em -cb}, {\em -cs}.
By default, \textsc{MiniSat+} uses the so called {\bf Mixed} strategy, where program
chooses which method ({\em -ca}, {\em -cb} or {\em -cs}) to use in the encodings, and this depends
on size of the input and memory restrictions. We don't include {\bf Mixed} strategy in the results,
as the evaluation showed that it performs almost the same as {\em -cb} option.
The generated CNFs were written to files with the option {\em -cnf=$<$file$>$}.

To sum up, here are the competitors' encodings used in this evaluation:

\begin{itemize}
  \item {\bf PCN} - the Pairwise Cardinality Networks (our implementation)
  \item {\bf CA} - encodings based on Binary Adders (from \textsc{MiniSat+})
  \item {\bf CB} - encodings based on Binary Decision Diagrams (from \textsc{MiniSat+})
  \item {\bf CS} - the Odd-Even Sorting Networks (from \textsc{MiniSat+})
\end{itemize}

Additionally, encodings {\bf 4WISE}, {\bf 4OE} and {\bf PCN} were extended,
following the idea presented in \cite{abio}, where authors use
Direct Cardinality Networks in their encodings for sufficiently
small values of $n$ and $k$. Values of $n$ and $k$ for which we substitute the recursive calls
with Direct Cardinality Network were selected based on our knowledge and some experiments.

Solver \textsc{MiniSat+} have been slightly modified, namely, we fixed several bugs such as
the one reported in the experiments section of \cite{aavani}.

\vspace{0.3cm}

\noindent {\bf Benchmarks:} We used the following sets of benchmarks:

\noindent {\bf 1. MSU4 suite:} A set consisting of about 14000 benchmarks, each of which
contains a mix of CNF formula and multiple cardinality constraints. This suite was
created from a set of MaxSAT instances reduced from real-life problems, and then it was 
converted by the implementation of {\em msu4} algorithm \cite{msu4}. This algorithm reduces
a MaxSAT problem to a series of SAT problems with cardinality constraints. The MaxSAT
instances were taken from the Partial Max-SAT division of the Third Max-SAT
evaluation\footnote{See {\texttt http://www.maxsat.udl.cat/08/index.php?disp=submitted-benchmarks}}.
We perform experiments for this benchmarks using \textsc{MiniSat}.

\noindent {\bf 2. PB15 suite:} A set of benchmarks from the
Pseudo-Boolean Evaluation 2015\footnote{See {\texttt http://pbeva.computational-logic.org/}}.
One of the categories of the competition was {\em DEC-LIN-32-CARD}, which contains 2289 instances
-- we use these in our evaluation. Every instance is a set of cardinality constraints. This set
of benchmarks were evaluated using \textsc{GlueMiniSat}. 

\renewcommand{\arraystretch}{1.1}
\begin{table}[t]\setlength{\tabcolsep}{4pt}
    \centering
    \begin{tabular}{ c | c | c | c | c | c | c || c | c | c | c | c | c | } \cline{2-13}
         & \multicolumn{6}{|c||}{4OE speed-up} & \multicolumn{6}{c|}{4OE slow-down} \\ 
         \cline{2-13}
                                  & TO  & 4.0  & 2.0 & 1.5 & 1.1 & Total & TO & 4.0 & 2.0 & 1.5 & 1.1 & Total \\ \hline
      \multicolumn{1}{|l|}{PCN}   & 54  & 73   & 1   & 6   & 9   & 143   & 1  & 6   & 0   & 2   & 4   & 13 \\ \hline
      \multicolumn{1}{|l|}{4WISE} & 26  & 48   & 11  & 5   & 4   & 94    & 2  & 7   & 4   & 3   & 4   & 20 \\ \hline
      \multicolumn{1}{|l|}{CA}    & 840 & 1114 & 0   & 6   & 4   & 1964  & 2  & 3   & 1   & 1   & 3   & 10 \\ \hline
      \multicolumn{1}{|l|}{CB}    & 100 & 187  & 0   & 7   & 3   & 297   & 3  & 6   & 1   & 4   & 2   & 16 \\ \hline
      \multicolumn{1}{|l|}{CS}    & 53  & 73   & 0   & 2   & 4   & 132   & 2  & 17  & 0   & 2   & 0   & 21 \\ \hline
      \multicolumn{13}{c}{ } \\
      \cline{2-13}
         & \multicolumn{6}{|c||}{4WISE speed-up} & \multicolumn{6}{c|}{4WISE slow-down} \\ 
         \cline{2-13}
                                & TO  & 4.0  & 2.0 & 1.5 & 1.1 & Total & TO  & 4.0 & 2.0 & 1.5 & 1.1 & Total \\ \hline
      \multicolumn{1}{|l|}{PCN} & 30  & 35   & 16  & 8   & 12  & 101   & 1   & 7   & 6   & 2   & 6   & 22 \\ \hline
      \multicolumn{1}{|l|}{4OE} & 2   & 7    & 4   & 3   & 4   & 20    & 26  & 48  & 11  & 5   & 4   & 94 \\ \hline
      \multicolumn{1}{|l|}{CA}  & 816 & 1083 & 13  & 7   & 8   & 1927  & 2   & 3   & 4   & 1   & 3   & 13 \\ \hline
      \multicolumn{1}{|l|}{CB}  & 86  & 154  & 23  & 11  & 6   & 280   & 13  & 18  & 5   & 2   & 4   & 42 \\ \hline
      \multicolumn{1}{|l|}{CS}  & 35  & 50   & 6   & 2   & 5   & 98    & 8   & 34  & 5   & 3   & 3   & 53 \\ \hline
      \multicolumn{13}{c}{ }
    \end{tabular}
    \caption{Comparison of encodings in terms of SAT-solver runtime on MSU suite. We count number of 
    benchmarks for which {\bf 4OE} and {\bf 4WISE} showed speed-up or slow-down factor with respect to 
    different encodings.}
    \label{tbl:res}
\end{table}

\begin{table}[t]\setlength{\tabcolsep}{4pt}
    \centering
    \begin{tabular}{ c | c | c | c | c | c | c || c | c | c | c | c | c | } \cline{2-13}
         & \multicolumn{6}{|c||}{4OE speed-up} & \multicolumn{6}{c|}{4OE slow-down} \\ 
         \cline{2-13}
                                  & TO & 4.0 & 2.0 & 1.5 & 1.1 & Total & TO & 4.0 & 2.0 & 1.5 & 1.1 & Total \\ \hline
      \multicolumn{1}{|l|}{PCN}   & 4  & 4   & 1   & 4   & 11  & 24    & 0  & 0   & 2   & 4   & 8   & 14 \\ \hline
      \multicolumn{1}{|l|}{4WISE} & 2  & 0   & 2   & 4   & 9   & 17    & 0  & 0   & 3   & 3   & 10  & 16 \\ \hline
      \multicolumn{1}{|l|}{CA}    & 10 & 13  & 32  & 13  & 20  & 88    & 2  & 2   & 6   & 5   & 15  & 30 \\ \hline
      \multicolumn{1}{|l|}{CB}    & 8  & 5   & 15  & 7   & 18  & 53    & 6  & 0   & 5   & 14  & 7   & 32 \\ \hline
      \multicolumn{1}{|l|}{CS}    & 15 & 10  & 17  & 14  & 19  & 75    & 6  & 1   & 5   & 8   & 17  & 37 \\ \hline
      \multicolumn{13}{c}{ } \\
    \end{tabular}
    \caption{Comparison of encodings in terms of SAT-solver runtime on PB15 suite. We count number of 
    benchmarks for which {\bf 4OE} showed speed-up or slow-down factor with respect to 
    different encodings.}
    \label{tbl:res-pb}
\end{table}

\begin{table}[t]\setlength{\tabcolsep}{4pt}
  \centering
  \begin{tabular}{ c | c | c | c | c | c | c | c | } \cline{3-8}
    \multicolumn{2}{c|}{}                                         & 4OE     & 4WISE   & CS      & PCN     & CB      & CA      \\ \hline
    \multicolumn{1}{|c|}{\multirow{2}{*}{\bf MSU4}}  & Total time & 25319   & 44978   & 52284   & 62695   & 102894  & 695003  \\ \cline{2-8}
    \multicolumn{1}{|c|}{}                           & Time-outs  & 16      & 40      & 67      & 69      & 113     & 854     \\ \hline
    \multicolumn{1}{|c|}{\multirow{2}{*}{\bf PB15}}  & Total time & 2291394 & 2292065 & 2311042 & 2294726 & 2301762 & 2312199 \\ \cline{2-8}
    \multicolumn{1}{|c|}{}                           & Time-outs  & 1240    & 1242    & 1249    & 1244    & 1242    & 1248    \\ \hline
    \multicolumn{8}{c}{ }
  \end{tabular}
  \caption{Total running time (in seconds) of all instances and total number of time-outs.}
  \label{tbl:total}
\end{table}

\vspace{0.3cm}

\noindent {\bf Results:}  The time-out limit in
the SAT-solvers was set to 600 seconds in case of MSU4 suite, and 1800 seconds for PB15 suite. When comparing two encodings we only considered instances
for which at least one achieved the SAT-solver runtime of at least 10\% of the time-out limit. All other instances
were considered trivial, and therefore were not included in the results. We also filtered out the instances for which
relative percentage deviation of the running time of encoding {\bf A} w.r.t the running time of encoding {\bf B}
was less than 10\% (and vice-versa), i.e., $|(t(A,I)-t(B,I))|/\min(t(A,I), t(B,I)) \leq 0.1$,
where $t(E,I)$ is the time in which the SAT-solver terminated given a CNF generated from the instance $I$ using encoding $E$.

Tables \ref{tbl:res} and \ref{tbl:res-pb} present speed-up and slow-down factors for
encodings {\bf 4OE} and {\bf 4WISE} w.r.t all other encodings.
For PB15 suite we only show speed-up/slow-down of {\bf 4OE},
since from the first batch of tests (on MSU4 suite) we can conclude that it is the superior encoding.
Explanation of how to read the tables now follows.
Consider the comparison of encoding {\bf A} w.r.t encoding {\bf B}.
In the speed-up section, {\bf TO} indicates the number of instances run with encoding {\bf B} that
ended as time-outs but were solved while run with encoding {\bf A}.
The {\bf TO} value in the slow-down section means that the opposite is true.
Column headers with numbers indicate speed-up (slow-down) factors
for the encoding {\bf A}, for example, instances counted in column $1.5$ are those
for which $2.0 \cdot t(A,I) \geq t(B,I) > 1.5 \cdot t(A,I)$ for speed-up,
and $2.0 \cdot t(A,I) \leq t(B,I) < 1.5 \cdot t(A,I)$ for slow-down.
In case of factor $4.0$ the other coefficient is $\infty$.

From the evaluation we can conclude that the best performing encoding is {\bf 4OE},
on both MSU4 and PB15 suites. On MSU4 suite (Table \ref{tbl:res}), the superiority is apparent, as our encoding
achieve significant speed-up factor w.r.t all other encodings. The runner-up encoding is {\bf 4WISE}, which performs
worse than {\bf 4OE}, but still beats all other encodings. This is also reflected in the total running times and number of time-outs
met over all instances -- those values are presented in Table \ref{tbl:total}. We see that on MSU4 suite,
SAT-solver finished twice as fast using encoding {\bf 4OE} rather than {\bf CS},
and solved almost all instances in given time. Similar observation can be made for {\bf 4WISE} w.r.t {\bf PCN}, where
our encoding has about 18 000 seconds (5 hours) smaller total time and 75\% less time-outs. This shows
that using 4-column selection networks is more desirable than using 2-column selection/sorting
networks for encoding cardinality constraints. Encodings {\bf CA} and {\bf CB} had the worst performance on MSU4 suite.

As for the results on PB15 suite, we see that the total running times and time-outs are similar (Table \ref{tbl:total}) for
all encodings. This is due to the fact that PB15 suite contains a lot of hard instances and a lot of trivial instances, so the true
competition boils down to a few medium-sized instances. Still, encoding {\bf 4OE} is the clear winner,
having finished computations about 20 000 seconds (5.5 hours) faster than {\bf CS}. Encoding {\bf 4WISE} performs similar
to {\bf 4OE}, with only about 10 minutes gap in the total running time, but it is the latter that solved the most number of instances.
Encoding {\bf 4OE} also outperforms all other encodings in terms of speed-up, which is presented in Table \ref{tbl:res-pb}.

In Figure \ref{fig:plots} we present two cactus plots, where x-axis gives the number of solved instances
and the y-axis the time needed to solve them (in seconds) using given encoding.
Results for MSU4 suite are in Figure \ref{fig:plot_cactus_1} and for PB15 suite in Figure \ref{fig:plot_cactus_2}.
Notice that we did not include the {\bf CA} encoding in the first plot, as it was not competitive to the other encodings.
From the plots we can see again that the {\bf 4OE} encoding outperforms all other encodings.

\begin{figure}
  \centering
  \subfloat[MSU4\label{fig:plot_cactus_1}]{%
    \includegraphics[width=0.48\textwidth]{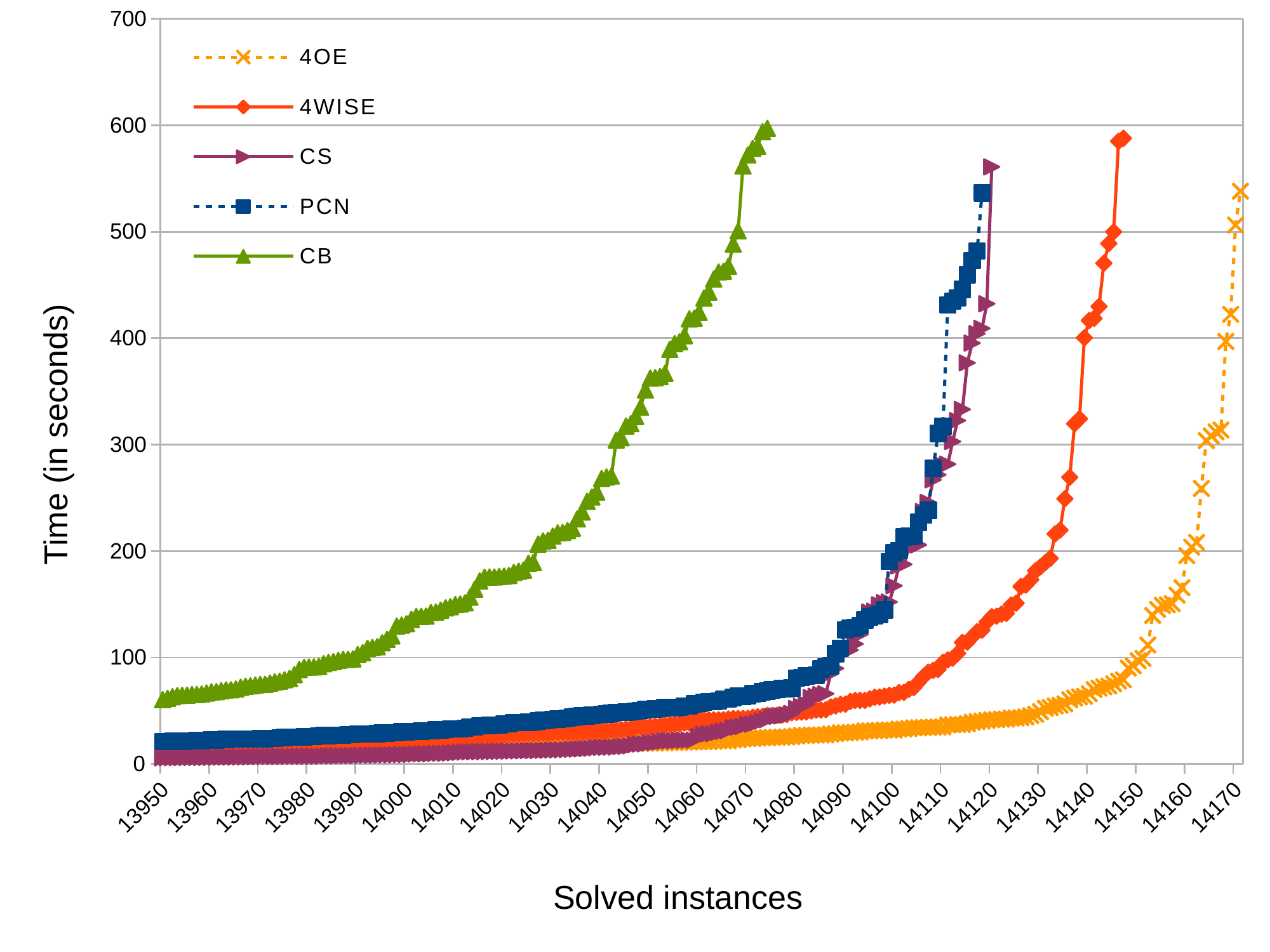}
  }
  ~
  \subfloat[PB15\label{fig:plot_cactus_2}]{%
    \includegraphics[width=0.48\textwidth]{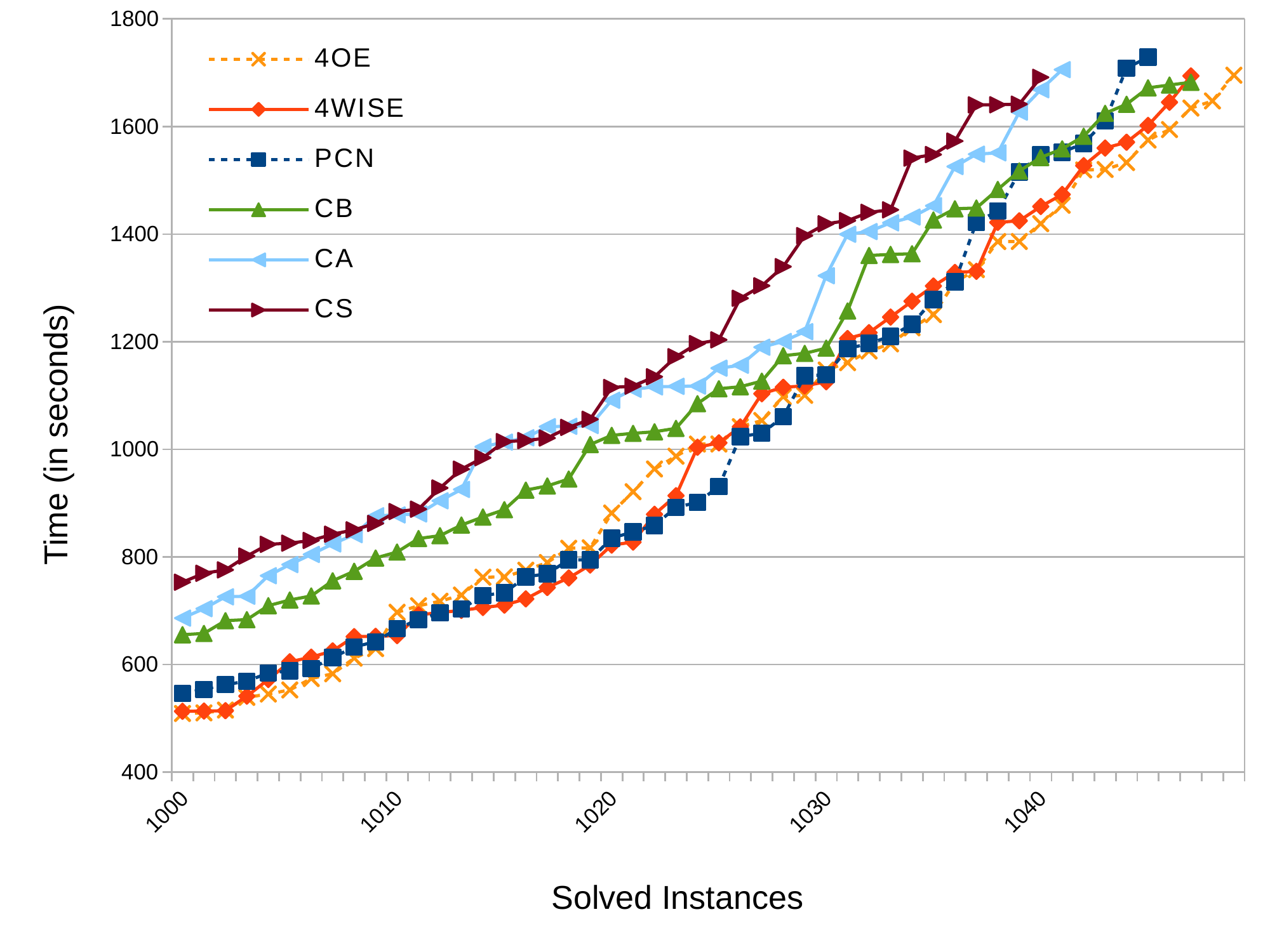}
  }
  \caption{Number of solved instances of MSU4 suite and PB15 suite in given time.}
  \label{fig:plots}
\end{figure}

\section{Conclusions} \label{sec:concl}

In this paper we presented a family of multicolumn selection networks based on odd-even and pairwise approach,
that can be used to encode cardinality constraints. We showed detailed
constructions where the number of columns is equal to 4 and showed that their CNF
encodings are smaller than their 2-column counterparts. We extended the encodings by
applying the Direct Cardinality Networks of \cite{abio} for sufficiently small input. The
new encodings were compared with the selected state-of-art encodings based
on comparator networks, adders and binary decision diagrams.
The experimental evaluation shows that the new encodings yield
significant speed-ups in the SAT-solver performance.

Developing new methods to encode cardinality constraints based on comparator networks is important
from the practical point of view. Using such encodings give an extra edge in solving optimization problems
for which we need to solve a series of problems that differ only in that a bound on cardinality constraint
$x_1 + \dots + x_n \leq k$ becomes tighter, i.e., by decreasing $k$ to $k'$. In this setting we only
need to add one more clause $\neg y_{k'}$, and the computation can be resumed while keeping all the previous
clauses untouched. This operation is allowed because if a comparator network is a $k$-selection network,
then it is also a $k'$-selection network, for any $k' < k$. This property is called {\em incremental strengthening}
and most state-of-art SAT-solvers provide an interface for doing this.

\end{document}